\documentclass[12pt, conference, onecolumn]{IEEEtran}

\usepackage{graphicx}
\usepackage{latexsym}
\usepackage{amssymb}
\usepackage{amsmath}
\usepackage{amsthm}
\usepackage{amsfonts}
\usepackage{xcolor}
\usepackage{multirow}
\usepackage{algorithm}
\usepackage{algorithmic}
\usepackage{enumerate}
\usepackage{caption}
\usepackage{mathtools}

\DeclareGraphicsExtensions{.tif,.pdf,.jpeg,.png,.jpg,.bmp,.svg,.eps,.EMF}

\newtheorem{theorem}{Theorem}
\newtheorem{definition}{Definition}
\newtheorem{cor}{Corollary}

\newcommand{\poly}{\mathrm{poly}}
\newcommand{\bX}{\mathbf{x}}
\newcommand{\bY}{\mathbf{y}}


\begin{document}

\title{Efficient (nonrandom) construction and decoding for non-adaptive group testing}

\author{
    \IEEEauthorblockN{Thach V. Bui\IEEEauthorrefmark{1}, Minoru Kuribayashi\IEEEauthorrefmark{2}, Tetsuya Kojima\IEEEauthorrefmark{3},\\ Roghayyeh Haghvirdinezhad\IEEEauthorrefmark{4}, and Isao Echizen\IEEEauthorrefmark{5}}
    \IEEEauthorblockA{\IEEEauthorrefmark{1} SOKENDAI (The Graduate University for Advanced Studies), Kanagawa, Japan. Email: bvthach@nii.ac.jp}
    \IEEEauthorblockA{\IEEEauthorrefmark{2} Okayama University, Japan. Email: kminoru@okayama-u.ac.jp}
    \IEEEauthorblockA{\IEEEauthorrefmark{3} National Institute of Technology, Tokyo College, Japan. Email: kojt@tokyo-ct.ac.jp}
    \IEEEauthorblockA{\IEEEauthorrefmark{4} New Jersey Institute of Technology, New Jersey, USA. Email: rh284@njit.edu}
    \IEEEauthorblockA{\IEEEauthorrefmark{5} National Institute of Informatics, Tokyo, Japan. Email: iechizen@nii.ac.jp}
}

\maketitle

\thispagestyle{plain}
\pagestyle{plain}

\begin{abstract}
The task of non-adaptive group testing is to identify up to $d$ defective items from $N$ items, where a test is positive if it contains at least one defective item, and negative otherwise. If there are $t$ tests, they can be represented as a $t \times N$ measurement matrix. We have answered the question of whether there exists a scheme such that a larger measurement matrix, built from a given $t\times N$ measurement matrix, can be used to identify up to $d$ defective items in time $O(t \log_2{N})$. In the meantime, a $t \times N$ nonrandom measurement matrix with $t = O \left(\frac{d^2 \log_2^2{N}}{(\log_2(d\log_2{N}) - \log_2{\log_2(d\log_2{N})})^2} \right)$ can be obtained to identify up to $d$ defective items in time $\poly(t)$. This is much better than the best well-known bound, $t = O \left( d^2 \log_2^2{N} \right)$. For the special case $d = 2$, there exists an efficient nonrandom construction in which at most two defective items can be identified in time $4\log_2^2{N}$ using $t = 4\log_2^2{N}$ tests. Numerical results show that our proposed scheme is more practical than existing ones, and experimental results confirm our theoretical analysis. In particular, up to $2^{7} = 128$ defective items can be identified in less than $16$s even for $N = 2^{100}$.
\end{abstract}

\section{Introduction}
\label{sec:intro}

Group testing dates back to World War II, when an economist, Robert Dorfman, solved the problem of identifying which draftees had syphilis~\cite{dorfman1943detection}. It turned out to a problem of finding up to $d$ defective items in a huge number of items $N$ by testing $t$ subsets of $N$ items. The meanings of ``items'', ``defective items'', and ``tests'' depend on the context. Classically, a test is positive if there is at least one defective item, and negative otherwise. Damaschke~\cite{damaschke2006threshold} generalized this problem into threshold group testing in which a test is positive if it contains at least $u$ defective items, negative if it contains at most $l$ defective items, and arbitrary otherwise. If $u = 1$ and $l = 0$, threshold group testing reduces to classical group testing.

In this work, we focus on classical group testing in which a test is positive if there exists at least one defective item, and negative otherwise. There are two main approaches to testing design: adaptive and non-adaptive. In \textit{adaptive group testing}, tests are performed in a sequence of stages, and the designs of later tests depend on the results of earlier tests. With this approach, the number of tests can be theoretically optimized~\cite{du2000combinatorial}. However, the testing can take a long time if there are many stages. Therefore, \textit{non-adaptive group testing} (NAGT)\cite{d1982bounds} is preferred: all tests are designed in advance and performed simultaneously. The growing use of NAGT in various fields such as compressed sensing~\cite{atia2012boolean}, data streaming~\cite{cormode2005s}, DNA library screening~\cite{ngo2000survey}, and neuroscience~\cite{bui2018a} has made it increasingly attractive recently. The focus here is thus on NAGT.

If $t$ tests are needed to identify up to $d$ defective items among $N$ items, they can be seen as a $t \times N$ measurement matrix. The procedure to get the matrix is called \textit{construction}, the procedure to get the outcome of $t$ tests using the measurement matrix is called \textit{encoding}, and the procedure to get the defective items from $t$ outcomes is called \textit{decoding}. Note that the encoding procedure includes the construction procedure. The objective of NAGT is to design a scheme such that all defective items are ``efficiently'' identified from the encoding and decoding procedures. Six criteria determine the efficiency of a scheme: measurement matrix construction type, number of tests needed, decoding time, time needed to generate an entry for the measurement matrix, space needed to generate a measurement matrix entry, and probability of successful decoding. The last criterion reduces the number of tests and/or the decoding complexity. With high probability, Cai et al.~\cite{cai2013grotesque} and Lee et al.~\cite{lee2016saffron} achieved a low number of tests and decoding complexity, namely $O(t)$, where $t = O (d\log{d} \cdot \log{N})$ ($\log$ is referred to as the logarithm of base 2). However, the construction type is random, and the whole measurement matrix must be stored for implementation, so it is limited to real-time applications. For example, in a data stream~\cite{cormode2005s}, routers have limited resources and need to be able to access the column in the measurement matrix assigned to an IP address as quickly as possible to perform their functions. The schemes proposed by Cai et al.~\cite{cai2013grotesque} and Lee et al.~\cite{lee2016saffron}, therefore, are inadequate for this application.

For exact identification of defective items, there are four main criteria to be considered: measurement matrix construction type, number of tests needed, decoding time, and time needed to generate measurement matrix entry. The measurement matrix is nonrandom if it always satisfies the preconditions after the construction procedure with probability $1$. It is random if it satisfies the preconditions after the construction procedure with some probability. A $t \times N$ measurement matrix is more practical if it is nonrandom, $t$ is small, the decoding time is a polynomial of $t$ ($\poly(t)$), and the time to generate its entry is also $\poly(t)$. However, there is always a trade-off between these criteria.

Kautz and Singleton~\cite{kautz1964nonrandom} proposed a scheme in which each entry in a $t \times N$ measurement matrix can be generated in $\poly(t)$, where $t = O(d^2 \log^2{N})$. However, the decoding time is $O(tN)$. Indyk et al.~\cite{indyk2010efficiently} reduced the decoding time to $\poly(t)$ while maintaining the order of the number of tests and the time to generate the entries. However, the number of tests in a nonrandom measurement matrix is not optimal.

In term of the pessimum number of tests, Guruswami and Indyk~\cite{guruswami2004linear} proposed a linear-time decoding scheme in accordance with the number of tests of $O(d^4 \log{N})$. To achieve an optimal bound on the number of tests, i.e., $O(d^2 \log{N})$, while maintaining a decoding time of $\poly(t)$ and keeping the entry computation time within $\poly(t)$, Indyk et al.~\cite{indyk2010efficiently} proposed a random construction. Although they tried to derandomize their schemes, it takes $\poly(t, N)$ time to construct such matrices, which is impractical when $d$ and $N$ are sufficiently large.

Cheraghchi~\cite{cheraghchi2013noise} achieved similar results. However, his proposed scheme can deal with the presence of noise in the test outcomes. Porat and Rothschild~\cite{porat2008explicit} showed that it is possible to construct a nonrandom $t \times N$ measurement matrix in time $O(tN)$ while maintaining the order of the number of tests, i.e., $O(d^2 \log{N})$. However, each entry in the resulting matrix is identified after the construction is completed. This is equivalent to each entry being generated in time $O(tN)$. If we reduce the number of tests, the nonrandom construction proposed by Indyk et al.~\cite{indyk2010efficiently} is the most practical.

\subsection{Contributions}
\label{sub:contri}

\textbf{Overview:} There are two main contributions in this work. First, we have answered the question of whether there exists a scheme such that a larger measurement matrix, built from a given $t\times N$ measurement matrix, can be used to identify up to $d$ defective items in time $O(t \log{N})$. Second, a $t \times N$ nonrandom measurement matrix with $t = O \left(\frac{d^2 \log^2{N}}{(\log(d\log{N}) - \log{\log(d\log{N})})^2} \right)$ can be obtained to identify up to $d$ defective items in time $\poly(t)$. This is much better than the best well-known bound $t = O \left( d^2 \log^2{N} \right)$. There is a special case for $d = 2$ in which there exists a $4\log^2{N} \times N$ nonrandom measurement matrix such that it can be used to identify up to two defective items in time $4\log^2{N}$. Numerical results show that our proposed scheme is the most practical and experimental results confirm our theoretical analysis. For instance, at most $2^{7} = 128$ defective items can be identified in less than $16$s even for $N = 2^{100}$.

\textbf{Comparison:} We compare variants of our proposed scheme with existing schemes in Table~\ref{tbl:cmp}. As mentioned above, six criteria determine the efficiency of a scheme: measurement matrix construction type, number of tests needed, decoding time, time needed to generate measurement matrix entry, space needed to generate a measurement matrix entry, and probability of successful decoding. Since the last criterion is only used to reduce the number of tests, it is not shown in the table. If the number of tests and the decoding time are the top priorities, the construction in $\langle \mathbf{11} \rangle$ is the best choice. However, since the probability of successful decoding is at least $1 - \epsilon$ for any $\epsilon > 0$, some defective items may not be identified.

From here on, we assume that the probability of successful decoding is 1; i.e., all defective items are identified. There are trade-offs among the first five criteria. When $d = 2$, the number of tests with our proposed scheme ($\langle \mathbf{8} \rangle$) is slightly larger than that with $\langle 7 \rangle$, although our proposed scheme has the best performance for the remaining criteria. When $d > 2$, the comparisons are as follows. First, if the generation of a measurement matrix must be certain, the best choices are $\langle 1 \rangle, \langle 2 \rangle, \langle \mathbf{3} \rangle, \langle \mathbf{4} \rangle, \langle 5 \rangle,$ and $\langle \mathbf{6} \rangle$. Second, if the number of tests must as low as possible, the best choices are $\langle 2 \rangle, \langle 5 \rangle,$ and $\langle 9 \rangle$. Third, if the decoding time is most important, the best choices are three variations of our proposed scheme: $\langle \mathbf{4} \rangle, \langle \mathbf{6} \rangle$, and $\langle \mathbf{10} \rangle$. Fourth, if the time needed to generate a measurement matrix entry is most important, the best choices are $\langle 1 \rangle, \langle 3 \rangle, \langle \mathbf{4} \rangle, \langle 7 \rangle, \langle 9 \rangle$ and $\langle \mathbf{10} \rangle$. Finally, if the space needed to generate a measurement matrix entry is most important, the best choices are $\langle 1 \rangle, \langle 2 \rangle, \langle \mathbf{3} \rangle, \langle \mathbf{4} \rangle, \langle 7 \rangle, \langle 9 \rangle$ and $\langle \mathbf{10} \rangle$.

For real-time applications, because ``defective items'' are usually considered to be \textit{abnormal system activities}~\cite{cormode2005s}, they should be identified as quickly as possible. It is thus acceptable to use extra tests to speed up their identification. Moreover, the measurement matrix deployed in the system should not be stored in the system because of saving space. Therefore, the construction type should be nonrandom, and the time and space needed to generate an entry should be within $\poly(t)$. Thus, the best choice is $\langle \mathbf{4} \rangle$ and the second best choice is $\langle \mathbf{3} \rangle$.

\begin{center}
\begin{table*}[ht]
\centering
\caption{Comparison with existing schemes.}
\scalebox{.87}{
\begin{tabular}{|c|c|c|c|c|c|c|}
\hline
No. & Scheme & \begin{tabular}{@{}c@{}}  Construction \\ type \end{tabular} & \begin{tabular}{@{}c@{}}  Number of tests \\ $t$  \end{tabular} & Decoding time & \begin{tabular}{@{}c@{}} Time to \\ generate \\ an entry \end{tabular} & \begin{tabular}{@{}c@{}} Space to \\ generate \\ an entry \end{tabular} \\
\hline
$\langle 1 \rangle$ & \begin{tabular}{@{}c@{}}  Indyk et al.~\cite{indyk2010efficiently} \\ (Theorem~\ref{thr:nonrandom}) \end{tabular} & Nonrandom & $O(d^2 \log^2{N})$ & $O \left( \frac{d^9 (\log{N})^{16 + 1/3}}{(\log(d \log{N}))^{7 + 1/3}}  \right) $ & $O(t)$ & $O(t)$ \\
\hline
$\langle 2 \rangle$ & \begin{tabular}{@{}c@{}}  Indyk et al.~\cite{indyk2010efficiently} \\ (Theorem~\ref{thr:StronglyExplicit}) \end{tabular} & Nonrandom & $O(d^2 \log{N})$ & $\poly(t) = O \left( d^{11} \log^{17}{N}  \right) $ & $\poly(t, N)$ & $\poly(t)$ \\
\hline
$\mathbf{\langle 3 \rangle}$ & \begin{tabular}{@{}c@{}}  \textbf{Proposed} \\ \textbf{(Theorem~\ref{thr:mainNonrandom})} \end{tabular} & Nonrandom & $O \left(\frac{d^2 \log^2{N}}{(\log(d\log{N}) - \log{\log(d\log{N})})^2} \right)$ & \begin{tabular}{@{}c@{}} $O \left( \frac{d^{3.57} \log^{6.26}{N}}{(\log(d\log{N}) - \log{\log(d\log{N})})^{6.26}} \right)$ \\ $+ O \left( \frac{d^6 \log^4{N}}{(\log(d\log{N}) - \log{\log(d\log{N})})^4} \right)$ \end{tabular} & $O(t)$ & $O(t)$ \\
\hline
$\mathbf{ \langle 4 \rangle }$ & \begin{tabular}{@{}c@{}} \textbf{Proposed} \\ \textbf{(Corollary~\ref{cor:nonrandom})} \end{tabular} & Nonrandom & $O \left(\frac{d^2 \log^3{N}}{(\log(d\log{N}) - \log{\log(d\log{N})})^2} \right)$ & $O(t)$ & $O(t)$ & $O(t)$ \\
\hline
$\langle 5 \rangle$ & \begin{tabular}{@{}c@{}} Porat-Rothschild~\cite{porat2008explicit} \\ (Theorem~\ref{thr:WeaklyExplicit}) \end{tabular} & Nonrandom & $O(d^2 \log{N})$ & $O(tN) = O(d^2 \log{N} \times N)$ & $O(tN)$ & $O(tN)$ \\
\hline
$\mathbf{ \langle 6 \rangle }$ & \begin{tabular}{@{}c@{}} \textbf{Proposed} \\ \textbf{(Corollary~\ref{cor:WeaklyExplicit})} \end{tabular} & Nonrandom & $O(d^2 \log^2{N})$ & $O(t) = O(d^2 \log^2{N})$ & \begin{tabular}{@{}c@{}} $O(tN)$ \end{tabular} & \begin{tabular}{@{}c@{}} $O(tN)$ \end{tabular} \\
\hline
$\langle 7 \rangle$ & \begin{tabular}{@{}c@{}}  Indyk et al.~\cite{indyk2010efficiently} \\ (Theorem~\ref{thr:nonrandom}) \end{tabular} & \begin{tabular}{@{}c@{}}  Nonrandom \\ $d = 2$ \end{tabular} & $2\log{N}(2\log{N} - 1)$ & $\frac{2^9 (\log{N})^{16 + 1/3}}{(\log(2 \log{N}))^{7 + 1/3}}$ & $\log^2{N}$ & \begin{tabular}{@{}c@{}} $\log{N}$ \end{tabular} \\
\hline
$\mathbf{ \langle 8 \rangle }$ & \begin{tabular}{@{}c@{}} \textbf{Proposed} \\ \textbf{(Theorem~\ref{thr:nonrandom2})} \end{tabular} & \begin{tabular}{@{}c@{}}  Nonrandom \\ $d = 2$ \end{tabular} & $4\log^2{N}$ & $4\log^2{N}$ & 4 & \begin{tabular}{@{}c@{}} $2\log{N}$ \\ $+ \log(2\log{N})$ \end{tabular} \\
\hline
$\langle 9 \rangle$ & \begin{tabular}{@{}c@{}} Indyk et al.~\cite{indyk2010efficiently} \\ (Theorem~\ref{thr:StronglyExplicit}) \end{tabular} & Random & $O(d^2 \log{N})$ & $\poly(t) = O \left( d^{11} \log^{17}{N}  \right) $ & $O(t^2 \log{N})$ & $O(t \log{N})$ \\
\hline
$\mathbf{ \langle 10 \rangle }$ & \begin{tabular}{@{}c@{}} \textbf{Proposed} \\ \textbf{(Corollary~\ref{cor:StronglyExplicit})} \end{tabular} & Random & $O(d^2 \log^2{N})$ & $O(t) = O(d^2 \log^2{N})$ & $O(t^2)$ & $O(t \log{N})$ \\
\hline
$\mathbf{ \langle 11 \rangle }$ & \begin{tabular}{@{}c@{}} \textbf{Proposed} \\ \textbf{(Corollary~\ref{cor:random})} \end{tabular} & Random & $O(d \log{N} \cdot \log{\frac{d}{\epsilon}})$ & $O(d \log{N} \cdot \log{\frac{d}{\epsilon}})$ & $O(tN)$ & $O(tN)$ \\
\hline
\end{tabular}}

\label{tbl:cmp}
\end{table*}
\end{center}

\subsection{Outline}
\label{sub:outline}
The paper is organized as follows. Section~\ref{sec:pre} presents some preliminaries on tensor product, disjunct matrices, list-recoverable codes, and a previous scheme. Section~\ref{sec:enlarge} describes how to achieve an efficient decoding scheme when a measurement matrix is given. Section~\ref{sec:nonrandom} presents nonrandom constructions for identifying up to two or more defective items. The numerical and experimental results are presented in Section~\ref{sec:eval}. The final section summarizes the key points and addresses several open problems.

\section{Preliminaries}
\label{sec:pre}

Notation is defined here for consistency. We use capital calligraphic letters for matrices, non-capital letters for scalars, and bold letters for vectors. Matrices and vectors are binary. The frequently used notations are as follows:
\begin{itemize}
\item $N; d$: number of items; maximum number of defective items. For simplicity, suppose that $N$ is the power of 2.
\item $|\cdot|$: weight; i.e, number of non-zero entries of input vector or cardinality of input set.
\item $\otimes, \circledcirc, \circ$: operation for NAGT, tensor product, concatenation code (to be defined later).
\item $\mathcal{S}, \mathcal{B}$: $k \times N$ measurement matrices used to identify at most one defective item, where $k = 2\log_2{N}$.
\item $\mathcal{M} = (m_{ij})$: $t \times N$ $d$-disjunct matrix, where integer $t \geq 1$ is number of tests.
\item $\mathcal{T} = (t_{ij})$: $T \times N$ measurement matrix used to identify at most $d$ defective items, where integer $T \geq 1$ is number of tests.
\item $\mathbf{x}; \mathbf{y}$: binary representation of $N$ items; binary representation of test outcomes.
\item $\mathcal{S}_j, \mathcal{B}_j, \mathcal{M}_j, \mathcal{M}_{i,*}$: column $j$ of matrix $\mathcal{S}$, column $j$ of matrix $\mathcal{B}$, column $j$ of matrix $\mathcal{M}$, row $i$ of matrix $\mathcal{M}$.
\item $\mathbb{D}$: index set of defective items, e.g., $\mathbb{D} = \{2, 6 \}$ means items 2 and 6 are defective.
\item $\mathrm{diag}(\mathcal{M}_{i, *}) = \mathrm{diag}(m_{i1}, m_{i2}, \ldots, m_{iN})$: diagonal matrix constructed by input vector $\mathcal{M}_{i, *} = (m_{i1}, m_{i2}, \ldots, m_{iN})$.
\item $\mathrm{e}, \log, \ln, \mathrm{exp}(\cdot)$: base of natural logarithm, logarithm of base 2, natural logarithm, exponential function.
\item $\lceil x \rceil, \lfloor x \rfloor$: ceiling and floor functions of $x$.
\end{itemize}

\subsection{Tensor product}
\label{sub:tensor}
Given an $f \times N$ matrix $\mathcal{A}$ and an $s \times N$ matrix $\mathcal{S}$, their tensor product $\circledcirc$ is defined as
\begin{align}
\label{tensor}
\mathcal{R} = \mathcal{A} \circledcirc \mathcal{S} &= \begin{bmatrix}
\mathcal{S} \times \mathrm{diag}(\mathcal{A}_{1, *}) \\
\vdots \\
\mathcal{S} \times \mathrm{diag}(\mathcal{A}_{f, *})
\end{bmatrix} \\
&= \begin{bmatrix}
a_{11} \mathcal{S}_1 & \ldots & a_{1N} \mathcal{S}_N \\
\vdots & \ddots & \vdots \\
a_{f1} \mathcal{S}_1 & \ldots & a_{fN} \mathcal{S}_N
\end{bmatrix}, 
\end{align}
where $\mathrm{diag}(.)$ is the diagonal matrix constructed by the input vector, $\mathcal{A}_{h, *} = (a_{h1}, \ldots, a_{hN})$ is the $h$th row of $\mathcal{A}$ for $h = 1, \ldots, f$, and $\mathcal{S}_j$ is the $j$th column of $\mathcal{S}$ for $j = 1, \ldots, N$. The size of $\mathcal{R}$ is $r \times N$, where $r = fs$. One can imagine that an entry $a_{hj}$ of matrix $\mathcal{A}$ would be replaced by the vector $a_{hj} \mathcal{S}_j$ after the tensor product is used. For instance, suppose that $f = 2, s = 3$, and $N = 4$. Matrices $\mathcal{A}$ and $\mathcal{S}$ are defined as
\begin{eqnarray}
\label{exampleAS}
\mathcal{A} = \begin{bmatrix}
1 & 0 & 1 & 0 \\
0 & 1 & 1 & 1
\end{bmatrix}, \quad
\mathcal{S} = \begin{bmatrix}
0 & 1 & 0 & 0 \\
1 & 0 & 1 & 1 \\
0 & 0 & 1 & 0
\end{bmatrix}.
\end{eqnarray}

Then $\mathcal{R} = \mathcal{A} \circledcirc \mathcal{S}$ is
\begin{align}
\label{exampleT}
\mathcal{R} &= \mathcal{A} \circledcirc \mathcal{S} = \begin{bmatrix}
1 & 0 & 1 & 0 \\
0 & 1 & 1 & 1
\end{bmatrix} \circledcirc
\begin{bmatrix}
0 & 1 & 0 & 0 \\
1 & 0 & 1 & 1 \\
0 & 0 & 1 & 0
\end{bmatrix} \\
&=
\begin{bmatrix}
1 \times \begin{bmatrix} 0 \\ 1 \\ 0 \end{bmatrix} & 0 \times \begin{bmatrix} 1 \\ 0 \\ 0 \end{bmatrix} & 1 \times \begin{bmatrix} 0 \\ 1 \\ 1 \end{bmatrix} & 0 \times \begin{bmatrix} 0 \\ 1 \\ 0 \end{bmatrix} \\ \\
0 \times \begin{bmatrix} 0 \\ 1 \\ 0 \end{bmatrix} & 1 \times \begin{bmatrix} 1 \\ 0 \\ 0 \end{bmatrix} & 1 \times \begin{bmatrix} 0 \\ 1 \\ 1 \end{bmatrix} & 1 \times \begin{bmatrix} 0 \\ 1 \\ 0 \end{bmatrix}
\end{bmatrix} \\
&= \begin{bmatrix}
0 & 0 & 0 & 0 \\
1 & 0 & 1 & 0 \\
0 & 0 & 1 & 0 \\
0 & 1 & 0 & 0 \\
0 & 0 & 1 & 1 \\
0 & 0 & 1 & 0
\end{bmatrix}.
\end{align}

\subsection{Disjunct matrices}
\label{sub:disjunct}

To gain insight into disjunct matrices, we present the concept of an identity matrix inside a set of vectors. This concept is used to later construct a $d$-disjunct matrix.
\begin{definition}
\label{def:Iden}
Any $c$ column vectors with the same size contain a $c \times c$ identity matrix if a $c \times c$ identity matrix could be obtained by placing those columns in an appropriate order.
\end{definition}

Note that there may be more than one identity matrix inside those $c$ vectors. For example, let $\mathbf{b}_1, \mathbf{b}_2$, and $\mathbf{b}_3$ be vectors of size $4 \times 1$:
\begin{eqnarray}
\mathbf{b}_1 = \begin{bmatrix}
1 \\
0 \\
0 \\
1 
\end{bmatrix}, \mathbf{b}_2 = \begin{bmatrix}
1 \\
1 \\
0 \\
0 
\end{bmatrix}, \mathbf{b}_3 = \begin{bmatrix}
1 \\
0 \\
1 \\
1 
\end{bmatrix}.
\end{eqnarray}
Then, $(\mathbf{b}_1, \mathbf{b}_2)$ and $(\mathbf{b}_2, \mathbf{b}_3)$ contain $2 \times 2$ identify matrices, whereas $(\mathbf{b}_1, \mathbf{b}_3)$ does not.
\begin{align}
\begin{bmatrix}
\mathbf{b}_1 & \mathbf{b}_2
\end{bmatrix} &= \begin{bmatrix}
1 & 1 \\
\textcolor{blue}{0} & \textcolor{blue}{1} \\
0 & 0 \\
\textcolor{blue}{1} & \textcolor{blue}{0}
\end{bmatrix}, 
\begin{bmatrix}
\mathbf{b}_2 & \mathbf{b}_3
\end{bmatrix} &= \begin{bmatrix}
1 & 1 \\
\textcolor{blue}{1} & \textcolor{blue}{0} \\
\textcolor{blue}{0} & \textcolor{blue}{1} \\
0 & 1
\end{bmatrix} = \begin{bmatrix}
1 & 1 \\
\textcolor{blue}{1} & \textcolor{blue}{0} \\
0 & 1 \\
\textcolor{blue}{0} & \textcolor{blue}{1}
\end{bmatrix}. \nonumber
\end{align}

The union of $l$ vectors is defined as follows. Given $l$ binary vectors $\mathbf{y}_w = (y_{1w}, y_{2w}, \ldots, y_{Bw})^T$ for $w=1, \ldots, l$ and some integer $B \geq 1$, their union is defined as vector $\mathbf{y} = \vee_{i = 1}^l \mathbf{y}_i = (\vee_{i = 1}^l y_{1i}, \ldots, \vee_{i = 1}^l y_{Bi})^T$, where $\vee$ is the OR operator.

Definition~\ref{def:Iden} is interchangeably defined as follows: the union of at most $c-1$ vectors does not contain the remaining vector. Here we use definition~\ref{def:Iden}, so the definition for a $d$-disjunct matrix is as follows.
\begin{definition}
\label{def:disjunct}
A binary $t \times N$ matrix is called a $d$-disjunct matrix iff there exists an $(d+1) \times (d+1)$ identity matrix in a set of $d + 1$ columns arbitrarily selected from the matrix.
\end{definition}

For example, a $3 \times 3$ identity matrix is a 2-disjunct matrix. The encoding and decoding procedures used to identify up to $d$ defective items using a $d$-disjunct matrix are as follows. Suppose that $\mathcal{M} = (m_{ij})$ is a $t \times N$ measurement matrix, which is used to identify at most $d$ defective items. Item $j$ is represented by column $\mathcal{M}_j$ for $j = 1, \ldots, N$. Test $i$ is represented by row $i$ in which $m_{ij} = 1$ iff the item $j$ belongs to test $i$, and $m_{ij} = 0$ otherwise, where $i = 1, \ldots, t$. Usually, $\mathcal{M}$ is a $d$-disjunct matrix, but this is not a requirement. In Section~\ref{sec:enlarge}, we will see that $\mathcal{M}$ may not be $d$-disjunct and still be able to to identify up to $d$ defective items.

Let $\mathbf{x} = (x_1, \ldots, x_N)^T$ be a binary representation for $N$ items, in which $x_j = 1$ iff item $j$ is defective for $j = 1, \ldots, N$. The outcome of $t$ tests, denoted as $\mathbf{y} = (y_1, \ldots, y_t)^T \in \{0, 1 \}^t$, is:
\begin{eqnarray}
\mathbf{y} = \mathcal{M} \otimes \mathbf{x} = \bigvee_{j = 1}^N x_j \mathcal{M}_j = \bigvee_{j \in \mathbb{D}} \mathcal{M}_j, \label{encM}
\end{eqnarray}
where $\mathbb{D}$ is the index set of defective items. The construction procedure is used to get $\mathcal{M}$. The encoding procedure (which includes the construction procedure) is used to get $\mathbf{y}$. The decoding procedure is used to recover $\mathbf{x}$ from $\mathbf{y}$ and $\mathcal{M}$.

We next present some recent results for the construction and decoding of disjunct matrices. With naive decoding, all items belonging to tests with negative outcomes are removed; the items remaining are considered to be defective. The decoding complexity of this approach is $O(tN)$. Naive decoding is used only a little here because the decoding time is long. A matrix is said to be nonrandom if its columns are deterministically generated without using randomness. In contrast, a matrix is said to be random if its columns are randomly generated. We thus classify construction types on the basis of the time it takes to generate a matrix entry. A $t \times N$ matrix is said to be weakly explicit if each of its columns is generated in time (and space) $O(tN)$. It is said to be strongly explicit if each of its columns is generated in time (and space) $\poly(t)$. We first present a weakly explicit construction of a disjunct matrix.

\begin{theorem}[Theorem 1~\cite{porat2008explicit}]
\label{thr:WeaklyExplicit}
Given $1 \leq d < N$, there exists a nonrandom $t \times N$ $d$-disjunct matrix that can be constructed in time $O(tN)$, where $t = O(d^2 \log{N})$. Moreover, the decoding time is $O(tN)$, and each column is generated in time (and space) $O(tN)$.
\end{theorem}

The second construction is strongly explicit.
\begin{theorem}[Corollary 5.1~\cite{indyk2010efficiently}]
\label{thr:StronglyExplicit}
Given $1 \leq d < N$, there exists a random $t \times N$ $d$-disjunct matrix that can be decoded in time $\poly(t) = O(d^{11} \log^{17}{N})$, where $t = 4800d^2 \log{N} = O(d^2 \log{N})$. Each column can be generated in time $O(t^2 \log{N})$ and space $O(t \log{N})$. There also exists a matrix that can be nonrandomly constructed in time $\poly(t, N)$ and space $\poly(t)$ while the construction time and space for each column of the matrix remain same.
\end{theorem}

Finally, the last construction is nonrandom. We analyze this construction in detail for later comparison. Although the precise formulas were not explicitly given in~\cite{indyk2010efficiently}, they can be derived.
\begin{theorem}[Corollary C.1~\cite{indyk2010efficiently}]
\label{thr:nonrandom}
Given $1 \leq d < N$, a nonrandom $t \times N$ $d$-disjunct matrix can be decoded in time $O \left( \frac{d^9 (\log{N})^{16 + 1/3}}{(\log(d \log{N}))^{7 + 1/3}}  \right) = \poly(t)$, where $t = O(d^2 \log^2{N})$. Moreover, each entry (column) can be generated in time (and space) $O(t)$ $(O(t^{3/2})).$ When $d = 2$, the number of tests is $2 \log{N} \times (2\log{N} -1)$, the decoding time is longer than $\frac{2^9 (\log{N})^{16 + 1/3}}{(\log(2 \log{N}))^{7 + 1/3}}$, and each entry is generated in time $\log^2{N}$ and space $\log{N}$.
\end{theorem}

\subsection{List recoverable codes}
\label{sub:LRCs}

There may be occasions in the physical world where a person might want to recover a similar codeword from a given codeword. For example, a person searching on a website such as Google might be searching using the word ``intercept''. However, mistyping results in the input word being ``inrercep''. The website should suggest a \textit{list} of similar words that are ``close'' to the input word such as ``in\textbf{t}ercep\textbf{t}'' and ``in\textbf{t}erce\textbf{de}''. This observation leads to the concept of \textit{list-recoverable codes}. The basic idea of list-recoverable codes is that, given a list of subsets in which each subset contains at most $\ell$ symbols in a given alphabet $\Sigma$ (a finite field), the decoder of the list-recoverable codes produces at most $L$ codewords from the list. Formally, this can be defined as follows.

\begin{definition}[Definition 2.2~\cite{guruswami2007algorithmic}]
Given integers $1 \leq \ell \leq L$, a code $C \subseteq \Sigma^n$ is said to be $(\ell, L)$-list-recoverable if for all sequences of subsets $S_1, S_2, \ldots, S_n$ with each $S_a \subset \Sigma$ satisfying $|S_a| \leq \ell$, there are at most $L$ codewords $\mathbf{c} = (c_1, \ldots, c_n) \in C$ with the property that $c_a \in S_a$ for $a \in \{1, 2, \ldots, n \}$. The value $\ell$ is referred to as the input list size.
\end{definition}

Note that for any $\ell^\prime \leq \ell$, an $(\ell, L)$-list-recoverable code is also an $(\ell^\prime, L)$-list-recoverable code. For example, if we set $\Sigma = \{a, b, \ldots, z \}, \ell = 2, n = 9$, and $L = 2$, we have the following input and output:
\begin{equation}
\begin{bmatrix}
S_1 = \{e, g \} \\
S_2 = \{r, x \} \\
S_3 = \{o, q \} \\
S_4 = \{t, u \} \\
S_5 = \{e, i \} \\
S_6 = \{s \} \\
S_7 = \{i, q \} \\
S_8 = \{t, u \} \\
S_9 = \{e \} \\
\end{bmatrix}
\xRightarrow{\text{decode}}
\mathbf{c} = \left\{
\begin{bmatrix}
e \\
x \\
q \\
u \\
i \\
s \\
i \\
t \\
e
\end{bmatrix}, \begin{bmatrix}
g \\
r \\
o \\
t \\
e \\
s \\
q \\
u \\
e 
\end{bmatrix}
\right\}. \nonumber
\end{equation}

\subsection{Reed-Solomon codes}
\label{sub:RS}
We first review the concept of $(n, r, D)_q$ code $C$:
\begin{definition}
Let $n, r, D, q$ be positive integers. An $(n, r, D)_q$ code is a subset of $\Sigma^n$ such that
\begin{enumerate}
\item $\Sigma$ is a finite field and is called the alphabet of the code: $|\Sigma| = q$. Here we set $\Sigma = \mathbb{F}_q$.
\item Each codeword is considered to be a vector of $\mathbb{F}_q^{n \times 1}$.
\item $D = \underset{\bX, \bY \in C}{\min} \Delta(\bX, \bY)$, where $\Delta(\bX, \bY)$ is the number of positions in which the corresponding entries of $\bX$ and $\bY$ differ.
\item The cardinality of $C$, i.e., $|C|$, is at least $q^r$.
\end{enumerate}
\end{definition}

These parameters $(n, r, D, q)$ are the the block length, dimension, minimum distance, and alphabet size of $C$. If the minimum distance is not considered, we refer to $C$ as $(n, r)_q$. Given a full-rank $n \times r$ matrix $\mathcal{G} \in \mathbb{F}_q^{n \times r}$, suppose that, for any $\mathbf{y} \in C$, there exists a message $\mathbf{x} \in \mathbb{F}_q^r$ such that $\mathbf{y} = \mathcal{G} \mathbf{x}$. In this case, $C$ is called a linear code and denoted as $[n, r, D]_q$. Let $\mathcal{M}_C$ denote an $n \times q^r$ matrix in which the columns are the codewords in $C$.

Reed-Solomon (RS) codes are constructed by applying a polynomial method to a finite field $\mathbb{F}_q$. Here we overview a common and widely used Reed-Solomon code, an $[n, r, D]_q$-code $C$ in which $|C| = q^r$ and $D = n - r + 1$. Since $D$ is determined from $n$ and $r$, we refer to $[n, r, D]_q$-RS code as $[n, r]_q$-RS code. Guruswami~\cite{guruswami2007algorithmic} (Section 4.4.1) showed that any $[n, r]_q$-RS code is also an $\left( \left\lceil \frac{n}{r} \right\rceil - 1, O\left( \frac{n^4}{r^2} \right) \right)$-list-recoverable code. To efficiently decode RS code, Chowdhury et al.~\cite{chowdhury2015faster} proposed an efficient scheme, which they summarized in Table 1 of their paper with $\omega < 2.38$, as follows:
\begin{theorem}[Corollary 18~\cite{chowdhury2015faster}]
\label{thr:decodeListRS}
Let $1 \leq r \leq n \leq q$ be integers. Then, any $[n, r]_q$-RS code, which is also $\left( \left\lceil \frac{n}{r} \right\rceil - 1, O\left( \frac{n^4}{r^2} \right) \right)$-list-recoverable code, can be decoded in time $O(n^{3.57} r^{2.69})$.
\end{theorem}
A codeword of the $[n, r]_q$-RS code can be computed in time $O(r^2 \log{\log{r}} ) \approx O(r^2)$ and space $O(r \log{q}/\log^2{r})$~\cite{von1997exponentiation}.

\subsection{Concatenated codes}
\label{sub:concate}
Concatenated codes $C$ are constructed by using an $(n_1, k_1)_q$ \textit{outer code} $C_{\mathrm{out}}$, where $q=2^{k_2}$ (in general, $q=p^{k_2}$ where $p$ is a prime number), and an $(n_2, k_2)_2$ binary \textit{inner code} $C_{\mathrm{in}}$, denoted as $C = C_{\mathrm{out}} \circ C_{\mathrm{in}}$.

Given a message $\mathbf{m} \in \mathbb{F}_q^{k_1}$, let $C_{\mathrm{out}}(\mathbf{m}) = (x_1,\ldots, x_{n_1}) \in \mathbb{F}_q^{n_1}$. Then $C_{\mathrm{out}} \circ C_{\mathrm{in}} (\mathbf{m}) = (C_{\mathrm{in}}(x_1), C_{\mathrm{in}}(x_2), \ldots, C_{\mathrm{in}}(x_{n_1})) \in ( \{0, 1  \}^{n_2})^{n_1}$. Note that $C$ is an $(n_1 n_2, k_1 k_2)_2$ code.

Using a suitable outer code and a suitable inner code, $d$-disjunct matrices can be generated. For example, let $C_{\mathrm{out}}$ and $C_{\mathrm{in}}$ be $(3, 1)_8$ and $(3, 3)_2$ codes, where $|C_{\mathrm{out}}| = 12$ and $|C_{\mathrm{in}}| = 8$. There corresponding matrices are $\mathcal{H} = \mathcal{M}_{C_{\mathrm{out}}}$ and $\mathcal{K} = \mathcal{M}_{C_{\mathrm{in}}}$ as follows:
\begin{align}
\mathcal{H} &= \left[ \begin{tabular}{cccccccccccc}
1 & 1 & 1 & 2 & 2 & 2 & 4 & 4 & 4 & 7 & 0 & 0 \\
1 & 2 & 4 & 1 & 2 & 4 & 1 & 2 & 4 & 0 & 7 & 0 \\
1 & 4 & 2 & 4 & 2 & 1 & 2 & 1 & 4 & 0 & 0 & 7 
\end{tabular}
\right], \notag \\
\mathcal{K} &= \left[
\begin{tabular}{cccccccc}
0 & 0 & 0 & 0 & 1 & 1 & 1 & 1 \\
0 & 0 & 1 & 1 & 0 & 0 & 1 & 1 \\
0 & 1 & 0 & 1 & 0 & 1 & 0 & 1
\end{tabular}
\right], \notag
\label{exampleK}
\end{align}

If we concatenate each element of $\mathcal{H}$ with its 3-bit binary representation such as matrix $\mathcal{K}$, we get a 2-disjunct matrix:
\begin{align}
\mathcal{M} &= \mathcal{H} \circ \mathcal{K} \nonumber\\
&= \left[
\begin{tabular}{cccccccccccc}
0 & 0 & 0 & 0 & 0 & 0 & 1 & 1 & 1 & 1 & 0 & 0 \\
0 & 0 & 0 & 1 & 1 & 1 & 0 & 0 & 0 & 1 & 0 & 0 \\
1 & 1 & 1 & 0 & 0 & 0 & 0 & 0 & 0 & 1 & 0 & 0 \\
0 & 0 & 1 & 0 & 0 & 1 & 0 & 0 & 1 & 0 & 1 & 0 \\
0 & 1 & 0 & 0 & 1 & 0 & 0 & 1 & 0 & 0 & 1 & 0 \\
1 & 0 & 0 & 1 & 0 & 0 & 1 & 0 & 0 & 0 & 1 & 0 \\
0 & 1 & 0 & 1 & 0 & 0 & 0 & 0 & 1 & 0 & 0 & 1 \\
0 & 0 & 1 & 0 & 1 & 0 & 1 & 0 & 0 & 0 & 0 & 1 \\
1 & 0 & 0 & 0 & 0 & 1 & 0 & 1 & 0 & 0 & 0 & 1 
\end{tabular}
\right] \nonumber
\end{align}

From this discussion, we can draw an important conclusion about decoding schemes using concatenation codes and list-recoverable codes.

\begin{theorem}[Simplified version of Theorem 4.1~\cite{indyk2010efficiently}]
\label{thr:conateCode}
Let $d, L \geq 1$ be integers. Let $C_{\mathrm{out}}$ be an $(n_1, k_1)_{2^{k_2}}$ code that can be $(d, L)$-list recovered in time $T_1(n_1, d, L, k_1, k_2)$. Let $C_{\mathrm{in}}$ be $(n_2, k_2)_2$ codes such that $\mathcal{M}_{C_{\mathrm{in}}}$ is a $d$-disjunct matrix that can be decoded in time $T_2(n_2, d, k_2)$. Suppose that matrix $\mathcal{M} = \mathcal{M}_{C_{\mathrm{out}} \circ C_{\mathrm{in}}}$ is $d$-disjunct. Note that $\mathcal{M}$ is a $t \times N$ matrix where $t = n_1 n_2$ and $N = 2^{k_1 k_2}$. Further, suppose that any arbitrary position in any codeword in $C_{\mathrm{out}}$ and $C_{\mathrm{in}}$ can be computed in space $S_1(n_1, d, L, k_1, k_2)$ and $S_2(n_2, d, k_2)$, respectively. Then:
\begin{enumerate}[(a)]
\item given any outcome produced by at most $d$ positives, the positive positions can be recovered in time $n_1 T_2(n_2, d, k_2) + T_1(n_1, d, L, k_1, k_2) + 2Lt = n_1 T_2(n_2, d, k_2) + T_1(n_1, d, L, k_1, k_2) + O(Lt)$; and
\item any entry in $\mathcal{M}$ can be computed in $\log{t} + \log{N} + S_1(n_1, d, L, k_1, k_2) + S_2(n_2, d, k_2) = O(\log{t} + \log{N}) + O \left( \max \{ S_1(n_1, d, L, k_1, k_2), S_2(n_2, d, k_2) \} \right) $ space.
\end{enumerate}
\end{theorem}

Since the decoding scheme requires knowledge from several fields that are beyond the scope of this work, we do not discuss it here. Readers are encouraged to refer to~\cite{indyk2010efficiently} for further reading.

\subsection{Review of Bui et al.'s scheme}
\label{sub:BuiScheme}
A scheme proposed by Bui et al.~\cite{bui2018efficient} plays an important role for constructions in later sections. it is used to identify at most one defective item while never producing a false positive. The technical details are as follows.

\textit{Encoding procedure:} Lee et al.~\cite{lee2016saffron} proposed a $k \times N$ measurement matrix $\mathcal{S}$ that uses $\log{N}$-bit representation of an integer, to detect at most one defective item:
\begin{equation}
\label{matrixS}
\mathcal{S} := \begin{bmatrix}
\mathbf{b}_1 & \mathbf{b}_2 \ldots \mathbf{b}_N \\
\overline{\mathbf{b}}_1 & \overline{\mathbf{b}}_2 \ldots \overline{\mathbf{b}}_N
\end{bmatrix} =
\begin{bmatrix}
\mathcal{S}_1 \ldots \mathcal{S}_N
\end{bmatrix},
\end{equation}
where $k = 2\log{N}$, $\mathbf{b}_j$ is the $\log{N}$-bit binary representation of integer $j-1$, $\overline{\mathbf{b}}_j$ is $\mathbf{b}_j$'s complement, and $\mathcal{S}_j := \begin{bmatrix} \mathbf{b}_j \\ \overline{\mathbf{b}}_j \end{bmatrix}$ for $j = 1,2,\ldots, N$. The weight of every column in $\mathcal{S}$ is $k/2 = \log{N}$.

Given an input vector $\mathbf{g} = (g_1, \ldots, g_N) \in \{0, 1 \}^N$, measurement matrix $\mathcal{S}$ is generalized:
\begin{equation}
\label{matrixB}
\mathcal{B} := \mathcal{S} \times \mathrm{diag}(\mathbf{g}) = \begin{bmatrix}
g_{1}\mathcal{S}_1 & \ldots & g_{N}\mathcal{S}_N
\end{bmatrix},
\end{equation}
where $\mathrm{diag}(\mathbf{g}) = \mathrm{diag}(g_{1}, \ldots, g_{N})$ is the diagonal matrix constructed by input vector $\mathbf{g}$, and $\mathcal{B}_{j} = g_j\mathcal{S}_{j}$ for $j = 1, \ldots, N$. It is obvious that $\mathcal{B} = \mathcal{S}$ when $\mathbf{g}$ is a vector of all ones; i.e., $\mathbf{g} = \mathbf{1} = (1, 1, \ldots, 1) \in \{ 1 \}^N$. Moreover, the column weight of $\mathcal{B}$ is either $k/2 = \log{N}$ or 0.

For example, consider the case $N = 8, k = 2\log{N} = 6$, and $\mathbf{g} = (1, \textcolor{blue}{0}, 1, \textcolor{blue}{0}, 1, 1, 1, 1)$. Measurement matrices $S$ and $\mathcal{B}$ are
\begin{eqnarray}
\mathcal{S} &=& 
\begin{bmatrix}
0 & 0 & 0 & 0 & 1 & 1 & 1 & 1 \\
0 & 0 & 1 & 1 & 0 & 0 & 1 & 1 \\
0 & 1 & 0 & 1 & 0 & 1 & 0 & 1 \\
1 & 1 & 1 & 1 & 0 & 0 & 0 & 0 \\
1 & 1 & 0 & 0 & 1 & 1 & 0 & 0 \\
1 & 0 & 1 & 0 & 1 & 0 & 1 & 0 \\
\end{bmatrix}, \label{exampleS} \\
\mathcal{B} &=& \mathcal{S} \times \mathrm{diag}(\mathbf{g}) = \mathcal{S} \times \mathrm{diag}(1, \textcolor{blue}{0}, 1, \textcolor{blue}{0}, 1, 1, 1, 1) \nonumber \\
&=& [ 1 \times \mathcal{S}_1 \ \textcolor{blue}{0} \times \mathcal{S}_2 \ 1 \times \mathcal{S}_3 \ \textcolor{blue}{0} \times \mathcal{S}_4 \  1 \times \mathcal{S}_5 \ 1 \times \mathcal{S}_6 \ 1 \times \mathcal{S}_7 \ 1 \times \mathcal{S}_8 ] \nonumber \\
&=& \begin{bmatrix}
0 & \textcolor{blue}{0} & 0 & \textcolor{blue}{0} & 1 & 1 & 1 & 1 \\
0 & \textcolor{blue}{0} & 1 & \textcolor{blue}{0} & 0 & 0 & 1 & 1 \\
0 & \textcolor{blue}{0} & 0 & \textcolor{blue}{0} & 0 & 1 & 0 & 1 \\
1 & \textcolor{blue}{0} & 1 & \textcolor{blue}{0} & 0 & 0 & 0 & 0 \\
1 & \textcolor{blue}{0} & 0 & \textcolor{blue}{0} & 1 & 1 & 0 & 0 \\
1 & \textcolor{blue}{0} & 1 & \textcolor{blue}{0} & 1 & 0 & 1 & 0 \\
\end{bmatrix}. \label{exampleB}
\end{eqnarray}

Then, given a representation vector of $N$ items $\mathbf{x} = (x_1, \ldots, x_N)^T \in \{ 0, 1\}^N$, the outcome vector is
\begin{eqnarray}
\mathbf{y}^\prime &=& \mathcal{B} \otimes \mathbf{x} = \bigvee_{j = 1}^N x_j \mathcal{B}_j \label{yi} \\
&=& \bigvee_{j = 1}^N x_j g_{j} \mathcal{S}_j = \bigvee_{\substack{j = 1 \\ x_j g_j = 1}}^N \mathcal{S}_j.
\label{OutcomeBui}
\end{eqnarray}
Note that, even if there is only one entry $x_{j_0} = 1$ in $\mathbf{x}$, index $j_0$ cannot be recovered if $g_{j_0} = 0$.

\textit{Decoding procedure:} From equation (\ref{OutcomeBui}), the outcome $\mathbf{y}^\prime$ is the union of at most $|\mathbf{x}|$ columns in $\mathcal{S}$. Because the weight of each column in $\mathcal{S}$ is $\log{N}$, if the weight of $\mathbf{y}^\prime$ is $\log{N}$, the index of one non-zero entry in $\mathbf{x}$ is recovered by checking the first half of $\mathbf{y}^\prime$. On the other hand, if $\mathbf{y}^\prime$ is the union of at least two columns in $\mathcal{S}$ or zero vector, the weight of $\mathbf{y}^\prime$ is not equal to $\log{N}$. This case is considered here as a defective item is not identified. Therefore, given a $k \times 1$ input vector, we can either identify one defective item or no defective item in time $ k = 2\log{N}= O(\log{N})$. Moreover, the decoding procedure does not produce a false positive.

For example, given $\mathbf{x}_1 = (0, 1, 0, 0, 0, 0, 0, 0)^T, \mathbf{x}_2 = (0, 1, 1, 0, 0, 0, 0, 0)^T$, and $\mathbf{x}_3 = (0, 1, 1, 1, 0, 0, 0, 0)^T$, their corresponding outcomes using the measurement matrix $\mathcal{B}$ in (\ref{exampleB}) are $\mathbf{y}^\prime_1 = (0, 0, 0, 0, 0, 0)^T, \mathbf{y}^\prime_2 = (0, 1, 0, 1, 0, 1)^T$, and $\mathbf{y}^\prime_3 = (0, 1, 0, 1, 0, 1)^T$. Since $|\mathbf{y}^\prime_1| = 0$, there is no defective item identified. Since $|\mathbf{y}^\prime_2| = |\mathbf{y}^\prime_3| = 3 = \log{N}$, the only defective item identified from the first half of $\mathbf{y}^\prime_2$ or $\mathbf{y}^\prime_3$, i.e., $(0, 1, 0)$ is 3. Note that, even if $|\mathbf{x}_1| \neq |\mathbf{x}_2|$, the same defective item is identified.

\section{Efficient decoding scheme using a given measurement matrix}
\label{sec:enlarge}
In this section, we present a simple but powerful tool for identifying defective items using a given measurement matrix. We thereby answer the question of whether there exists a scheme such that a larger $T \times N$ measurement matrix built from a given $t\times N$ measurement matrix, can be used to identify up to $d$ defective items in time $\poly(t) = t \times \log{N} = T$. It can be summarized as follows:

\begin{theorem}
\label{thr:mainTensor}
For any $\epsilon \geq 0$, suppose each set of $d$ columns in a given $t \times N$ matrix $\mathcal{M}$ contains a $d \times d$ identity matrix with probability at least $1 - \epsilon$. Then there exists a $T \times N$ matrix $\mathcal{T}$ constructed from $\mathcal{M}$ that can be used to identify at most $d$ defective items in time $T = t \times 2\log{N}$ with probability at least $1 - \epsilon$. Further, suppose that any entry of $\mathcal{M}$ can be computed in time $\beta$ and space $\gamma$, so every entry of $\mathcal{T}$ can be computed in time $O(\beta \log{N})$ and space $O(\log{T} + \log{N}) + O(\gamma \log{N})$.
\end{theorem}
\begin{proof}
Suppose $\mathcal{M} = (m_{ij}) \in \{0, 1 \}^{t \times N}$. Then the $T \times N$ measurement matrix $\mathcal{T}$ is generated by using the tensor product of $\mathcal{M}$ and $\mathcal{S}$ in (\ref{matrixS}):
\begin{eqnarray}
\mathcal{T} = \mathcal{M} \circledcirc \mathcal{S} &=& \begin{bmatrix}
\mathcal{S} \times \mathrm{diag}(\mathcal{M}_{1, *}) \\
\vdots \\
\mathcal{S} \times \mathrm{diag}(\mathcal{M}_{t, *}) \\
\end{bmatrix} = \begin{bmatrix}
\mathcal{B}^1 \\
\vdots \\
\mathcal{B}^t
\end{bmatrix} \notag \\
&=& \begin{bmatrix}
m_{11} \mathcal{S}_1 & \ldots & m_{1N} \mathcal{S}_N \\
\vdots & \ddots & \vdots \\
m_{t1} \mathcal{S}_1 & \ldots & m_{tN} \mathcal{S}_N
\end{bmatrix},
\end{eqnarray}
\noindent
where $T = t \times k = t \times 2\log{N}$ and $\mathcal{B}^i = \mathcal{S} \times \mathrm{diag}(\mathcal{M}_{i, *})$ for $i = 1, \ldots, t$. Note that $\mathcal{B}^i$ is an instantiation of $\mathcal{B}$ when $\mathbf{g}$ is set to $\mathcal{M}_{i, *}$ in (\ref{matrixB}). Then, for any $N \times 1$ representation vector $\mathbf{x} = (x_1, \ldots, x_N) \in \{0, 1 \}^N$, the outcome vector is
\begin{eqnarray}
\mathbf{y}^{\star} = \mathcal{T} \otimes \mathbf{x} = \begin{bmatrix}
\mathcal{B}^1 \otimes \mathbf{x} \\
\vdots \\
\mathcal{B}^t \otimes \mathbf{x}
\end{bmatrix} = \begin{bmatrix}
\mathbf{y}_1^\prime \\
\vdots \\
\mathbf{y}_t^\prime
\end{bmatrix},
\end{eqnarray}
\noindent
where $\mathbf{y}_i^\prime = \mathcal{B}^i \otimes \mathbf{x}$ for $i = 1, \ldots, t$; $\mathbf{y}_i^\prime$ is obtained by replacing $\mathcal{B}$ by $\mathcal{B}_i$ in (\ref{yi}).

By using the decoding procedure in section~\ref{sub:BuiScheme}, the decoding procedure is simply to can scan all $\mathbf{y}_i^\prime$ for $i = 1, \ldots, t$. If $|\mathbf{y}_i^\prime| = \log{N}$, we take the first half of $\mathbf{y}_i^\prime$ to calculate the defective item. Thus, the decoding complexity is $T = t \times 2\log{N} = O(T)$.

Our task now is to prove that the decoding procedure above can identify all defective items with probability at least $1 - \epsilon$. Let $\mathbb{D} = \{j_1, \ldots, j_{|\mathbb{D}|} \}$ be the defective set, where $|\mathbb{D}| = g \leq d$. We will prove that there exists $\mathbf{y}^\prime_{i_1}, \ldots, \mathbf{y}^\prime_{i_g}$ such that $j_a$ can be recovered from $\mathbf{y}_{i_a}^\prime$ for $a = 1, \ldots, g$. Because any set of $d$ columns in $\mathcal{M}$ contains a $d \times d$ identity matrix with probability at least $1 - \epsilon$, any set of $g \leq d$ columns $j_1, \ldots, j_g$ in $\mathcal{M}$ also contains a $g \times g$ identity matrix with probability at least $1 - \epsilon$. Let $i_1, \ldots, i_g$ be the row indexes of $\mathcal{M}$ such that $m_{i_a j_a} = 1$ and $m_{i_a j_b} = 0$, where $a, b \in \{1, 2, \ldots, g \}$ and $a \neq b$. Then the probability that rows $i_1, \ldots, i_g$ coexist is at least $1 - \epsilon$.

For any outcome $\mathbf{y}_{i_a}^\prime$, where $a = 1, \ldots, g$, by using (\ref{OutcomeBui}), we have
\begin{eqnarray}
\mathbf{y}_{i_a}^\prime = \mathcal{B}^{i_a} \otimes \mathbf{x} = \bigvee_{\substack{j = 1 \\ x_j m_{i_a j} = 1}}^N \mathcal{S}_j = \bigvee_{\substack{j \in \mathbb{D} \\ x_j m_{i_a j} = 1}} \mathcal{S}_j = \mathcal{S}_{j_a}, \ \label{getDefective}
\end{eqnarray}
\noindent
because there are only $g$ non-zero entries $x_{j_1}, \ldots, x_{j_g}$ in $\mathbf{x}$. Thus, all defective items $j_1, \ldots, j_g$ can be identified by checking the first half of each corresponding $\mathbf{y}_{i_1}^\prime, \ldots, \mathbf{y}_{i_g}^\prime$. Since the probability that rows $i_1, \ldots, i_g$ coexist is at least $1 - \epsilon$, the probability that defective items $j_1, \ldots, j_g$ are identified is also at least $1 - \epsilon$.

We next estimate the computational complexity of computing an entry in $\mathcal{T}$. An entry in row $1 \leq i \leq T$ and column $1 \leq j \leq N$ needs $\log{T} + \log{N}$ bits (space) to be indexed. It belongs to vector $m_{i_0 j} \mathcal{S}_j$, where $i_0 = i/(2 \log{N})$ if $i \mod (2\log{N}) \equiv 0$ and $i_0 = \lfloor i/(2 \log{N}) \rfloor$ if $i \mod (2\log{N}) \not\equiv 0$. Since each entry in $\mathcal{M}$ needs $\gamma$ space to compute, every entry in $\mathcal{T}$ can be computed in space $O(\log{T} + \log{N}) + O(\gamma \log{N})$ after mapping it to the corresponding column of $\mathcal{S}$. The time to generate an entry for $\mathcal{T}$ is straightforwardly obtained as $\beta \log{N} = O(\beta \log{N})$.
\end{proof}

Part of Theorem~\ref{thr:mainTensor} is implicit in other papers (e.g.,~\cite{bui2018efficient},~\cite{bui2017efficiently},~\cite{cai2013grotesque},~\cite{lee2016saffron}). However, the authors of those papers only considered cases specific to their problems. They mainly focused on how to generate matrix $\mathcal{M}$ by using complicated techniques and a non-constructive method, i.e., random construction (e.g.,~\cite{cai2013grotesque},~\cite{lee2016saffron}). As a result, their decoding schemes are randomized. Moreover, they did not consider the cost of computing an entry in $\mathcal{M}$. In two of the papers~\cite{bui2018efficient,bui2017efficiently}, the decoding time was not scaled to $t \times \log{N}$ for deterministic decoding, i.e., $\epsilon = 0$. Our contribution is to \textit{generalize} their ideas into the framework of non-adaptive group testing. We next instantiate Theorem~\ref{thr:mainTensor} in the broad range of measurement matrix construction.

\subsection{Case of $\epsilon = 0$}
\label{sub:epsi0}

We consider the case in which $\epsilon = 0$; i.e., a given matrix $\mathcal{M}$ is always $(d-1)$-disjunct. There are three metrics for evaluating an instantiation: number of tests, construction type, and time to generate an entry for $\mathcal{T}$. We first present an instantiation of a strongly explicit construction. Let $\mathcal{M}$ be a measurement matrix generated from Theorem~\ref{thr:StronglyExplicit}. Then $t = O(d^2 \log{N})$, $\beta = O(t^2 \log{N})$, and $\gamma = O(t \log{N})$. Thus, we obtain efficient NAGT where the number of tests and the decoding time are $O(d^2 \log^2{N})$.

\begin{cor}
\label{cor:StronglyExplicit}
Let $1 \leq d \leq N$ be integers. There exists a random $T \times N$ measurement matrix $\mathcal{T}$ with $T = O(d^2 \log^2{N})$ such that at most $d$ defective items can be identified in time $O(T)$. Moreover, each entry in $\mathcal{T}$ can be computed in time $O(T^2)$ and space $O(T \log{N})$.
\end{cor}

It is also possible to construct $\mathcal{T}$ deterministically. However, it would take $\poly(t, N)$ time and $\poly(t)$ space, which are too long and too much for practical applications. Therefore, we should increase the time needed to generate an entry for $\mathcal{T}$ in order to achieve nonrandom construction with the same number of tests $T = O(d^2 \log^2{N})$ and a short construction time. The following theorem is based on the weakly explicit construction of a given measurement matrix as in Theorem~\ref{thr:WeaklyExplicit}; i.e., $t = O(d^2 \log{N})$, $\beta = O(tN)$, and $\gamma = O(tN)$.

\begin{cor}
\label{cor:WeaklyExplicit}
Let $1 \leq d \leq N$ be integers. There exists a nonrandom $T \times N$ measurement matrix $\mathcal{T}$ with $T = O(d^2 \log^2{N})$ that can be used to identify at most $d$ defective items in time $O(T)$. Moreover, each entry in $\mathcal{T}$ can be computed in time (and space) $O(TN)$.
\end{cor}

Although the number of tests is low and the construction type is nonrandom, the time to generate an entry for $\mathcal{T}$ is long. If we increase the number of tests, one can achieve both nonrandom construction and low generating time for an entry as follows:

\begin{cor}
\label{cor:nonrandom}
Let $1 \leq d \leq N$ be integers. There exists a nonrandom $T \times N$ measurement matrix $\mathcal{T}$ with $T = O \left(\frac{d^2 \log^3{N}}{(\log(d\log{N}) - \log{\log(d\log{N})})^2} \right)$ that can be used to identify at most $d$ defective items in time $O(T)$. Moreover, each entry in $\mathcal{T}$ can be computed in time (and space) $O(T)$.
\end{cor}

The above corollary is obtained by choosing a measurement matrix as a $d$-disjunct matrix in Theorem~\ref{thr:mainNonrandom} (Section~\ref{sec:nonrandom}): $t = O \left(\frac{d^2 \log^2{N}}{(\log(d\log{N}) - \log{\log(d\log{N})})^2} \right)$, $\beta = O(t)$, and $\gamma = O(t)$.

\subsection{Case of $\epsilon > 0$}
\label{sub:epsiLarge0}
To reduce the number of tests and the decoding complexity, the construction process of the given measurement matrix must be randomized. We construct the matrix as follows. A given $t \times N$ matrix $\mathcal{M} = (m_{ij})$ is generated randomly, where $\Pr(m_{ij} = 1) = \frac{1}{d}$ and $\Pr(m_{ij} = 0) = 1 - \frac{1}{d}$ for $i = 1, \ldots, t$ and $j = 1, \ldots, N$. The value of $t$ is set to $\mathrm{e}d \ln{\frac{d}{\epsilon}}$. Then, for each set of $d$ columns in $\mathcal{M}$, the probability that a set does not contain a $d \times d$ identity matrix is at most
\begin{eqnarray}
&& \binom{d}{1} \left( 1 - \frac{1}{d} \left(1 - \frac{1}{d} \right)^{d-1} \right)^t \\
&\leq& d \cdot \mathrm{exp} \left( - \frac{1}{d-1} \left( 1 - \frac{1}{d} \right)^d t \right) \label{exp} \\
&\leq& d \cdot \mathrm{exp} \left( - \frac{t}{d-1} \cdot \mathrm{e}^{-1} \left( 1 - \frac{1}{d} \right) \right) \label{ineqnE} \\
&\leq& d \cdot \mathrm{exp} \left( - \frac{t}{\mathrm{e} d} \right) = d \cdot \mathrm{exp} \left( - \ln{\frac{d}{\epsilon}} \right) \\
&\leq& \epsilon.
\end{eqnarray}
Expression (\ref{exp}) is obtained because $(1 + x)^y \leq \mathrm{exp}(xy)$ for all $|x| \leq 1$ and $y \geq 1$. Expression (\ref{ineqnE}) is obtained because $\left(1 + \frac{x}{n} \right)^n \geq \mathrm{e}^x \left(1 - \frac{x^2}{n} \right)$ for $n > 1$ and $|x| < n$. Therefore, there exists a $t \times N$ matrix $\mathcal{M}$ with $t = O \left( d \log{\frac{d}{\epsilon}} \right)$ such that each set of $d$ columns contains a $d \times d$ identity matrix with probability at least $1 - \epsilon$, for any $\epsilon > 0$. Since $\beta= \gamma = O(tN)$, W can derive the following corollary.

\begin{cor}
\label{cor:random}
Given integers $1 \leq d \leq N$ and a scalar $\epsilon > 0$, there exists a random $T \times N$ measurement matrix $\mathcal{T}$ with $T = O \left( d \log{N} \cdot \log{\frac{d}{\epsilon}} \right)$ that can be used to identify at most $d$ defective items in time $O(T)$ with probability at least $1 - \epsilon$. Furthermore, each entry in $\mathcal{T}$ can be computed in time (and space) $O(TN)$.
\end{cor}

While the result in Corollary~\ref{cor:random} is similar to previously reported ones~\cite{cai2013grotesque},~\cite{lee2016saffron}, construction of matrix $\mathcal{M}$ is much simpler. It is possible to achieve the number of tests $t = O\left( d \log{\frac{d}{\epsilon}} \cdot \log{N} \right)$ when \textit{each} set of $d$ columns in $\mathcal{M}$ contains a $d \times d$ identity matrix with probability at least $1 - \epsilon$ for any $\epsilon > 0$. However, it is impossible to achieve this number for \textit{every} set of $d$ columns that contains a $d \times d$ identity matrix with probability at least $1 - \epsilon$. In this case, by using the same procedure used for generating random matrix $\mathcal{M}$ and by resolving $\binom{N}{d} \binom{d}{1} \left( 1 - \frac{1}{d} \left( 1 - \frac{1}{d} \right)^{d-1} \right)^t \leq \epsilon$, the number of tests needed is determined to be $t = O \left(d^2 \log{N} + d\log{\frac{1}{\epsilon}} \right)$. Since this number is greater than that when $\epsilon = 0$ ($O(d^2 \log{N})$), it is not beneficial to consider this case the case that every set of $d$ columns that contains a $d \times d$ identity matrix with probability at least $1 - \epsilon$.

\section{Nonrandom disjunct matrices}
\label{sec:nonrandom}
It is extremely important to have nonrandom constructions for measurement matrices in real-time applications. Therefore, we now focus on nonrandom constructions. We have shown that the well-known barrier on the number of tests $O(d^2 \log^2{N})$ for constructing a $d$-disjunct matrix can be overcome.

\subsection{Case of $d = 2$}
\label{sub:d2}
When $d = 2$, the measurement matrix is $\mathcal{T} = \mathcal{S} \circledcirc \mathcal{S}$, where $\mathcal{S}$ is given by (\ref{matrixS}). Note that the size of $\mathcal{S}$ is $k \times N$, where $k = 2\log{N}$, and $\mathcal{T}$ is not a 2-disjunct matrix. We start by proving that any two columns in $\mathcal{S}$ contain a $2 \times 2$ identity matrix. Indeed, suppose $\mathbf{b}_w = (b_{1w}, \ldots, b_{(k/2)w})^T$, which is a $\log{N}$-bit binary representation of $0 \leq w-1 \leq N-1$. For any two vectors $\mathbf{b}_{w_1}$ and $\mathbf{b}_{w_2}$, there exists a position $i_0$ such that $b_{i_0 w_1} = 0$ and $b_{i_0 w_2} = 1$, or $b_{i_0 w_1} = 1$ and $b_{i_0 w_2} = 0$ for any $1 \leq w_1 \neq w_2 \leq N$. Then their corresponding complementary vectors $\overline{\mathbf{b}}_{w_1} = (\overline{b}_{1w_1}, \ldots, \overline{b}_{(k/2)w_1})^T$ and $\overline{\mathbf{b}}_{w_2} = (\overline{b}_{1w_2}, \ldots, \overline{b}_{(k/2)w_2})^T$ satisfy: $\overline{b}_{i_0 w_1} = 0$ and $\overline{b}_{i_0 w_2} = 1$ when $b_{i_0 w_1} = 0$ and $b_{i_0 w_2} = 1$, or $\overline{b}_{i_0 w_1} = 1$ and $\overline{b}_{i_0 w_2} = 0$ when $b_{i_0 w_1} = 1$ and $b_{i_0 w_2} = 0$. Thus, any two columns $w_1$ and $w_2$ in $\mathcal{S}$ always contain a $2 \times 2$ identity matrix. From Theorem~\ref{thr:mainTensor} (set $\mathcal{M} = \mathcal{S}$), we obtain the following theorem.

\begin{theorem}
\label{thr:nonrandom2}
Let $2 \leq N$ be an integer. A $4\log^2{N} \times N$ nonrandom measurement matrix $\mathcal{T}$ can be used to identify at most two defective items in time $4\log^2{N}$. Moreover, each entry in $\mathcal{T}$ can be computed in space $2\log{N} + \log(2\log{N})$ with four operations.
\end{theorem}
\begin{proof}
It takes $\gamma = 2\log{N} + \log(2\log{N})$ bits to index an entry in row $i$ and column $j$. Only two shift operations and a mod operation are needed to exactly locate the position of the entry in column $\mathcal{S}_j$. Therefore, at most four operations ($\beta = 4$) and $2\log{N} + \log(2\log{N})$ bits are needed to locate an entry in matrix $\mathcal{T}$. The decoding time is straightforwardly obtained from Theorem~\ref{thr:mainTensor} ($t = k = 2\log{N}$).
\end{proof}

\subsection{General case}
\label{sub:d3}
Indyk et al.~\cite{indyk2010efficiently} used Theorem~\ref{thr:conateCode} and Parvaresh-Vardy (PV) codes~\cite{parvaresh2005correcting} to come up with Theorem~\ref{thr:nonrandom}. Since they wanted to convert RS code into list-recoverable code, they instantiated PV code into RS code. However, because PV code is powerful in terms of solving general problems, its decoding complexity is high. Therefore, the decoding complexity in Theorem~\ref{thr:nonrandom} is relatively high. Here, by converting RS code into list-recoverable code using Theorem~\ref{thr:decodeListRS}, we carefully use Theorem~\ref{thr:conateCode} to construct and decode disjunct matrices. As a result, the number of tests and the decoding time for a nonrandom disjunct matrix are significantly reduced.

Let $W(x)$ be a Lambert W function in which $W(x) \mathrm{e}^{W(x)} = x$ for any $x \geq -\frac{1}{\mathrm{e}}$. When $x > 0$, $W(x)$ is an increasing function. One useful bound~\cite{hoorfar2008inequalities} for a Lambert W function is $\ln{x} - \ln{\ln{x}} \leq W(x) \leq \ln{x} - \frac{1}{2} \ln{\ln{x}}$ for any $x \geq \mathrm{e}$. Theorem~\ref{thr:conateCode} is used to achieve the following theorem with careful setting of $C_{\mathrm{out}}$ and $C_{\mathrm{in}}$

\begin{theorem}
\label{thr:mainNonrandom}
Let $1 \leq d \leq N$ be integers. Then there exists a nonrandom $d$-disjunct matrix $\mathcal{M}$ with $t = O \left(\frac{d^2 \ln^2{N}}{(W(d\ln{N}))^2} \right) = O \left(\frac{d^2 \log^2{N}}{(\log(d\log{N}) - \log{\log(d\log{N})})^2} \right)$. Each entry (column) in $\mathcal{M}$ can be computed in time (and space) $O(t)$ $(O(t^{3/2})).$ Moreover, $\mathcal{M}$ can be used to identify up to $d^\prime$ defective items, where $d^\prime \geq \left\lfloor \frac{d}{2} \right\rfloor + 1$, in time
\begin{align}
O \left( \frac{d^{3.57} \log^{6.26}{N}}{(\log(d\log{N}) - \log{\log(d\log{N})})^{6.26}} \right) + O \left( \frac{d^6 \log^4{N}}{(\log(d\log{N}) - \log{\log(d\log{N})})^4} \right). \nonumber
\end{align}
When $d$ is the power of 2, $d^\prime = d - 1$.
\end{theorem}

\begin{proof}
\textbf{Construction:} We use the classical method proposed by Kautz and Singleton~\cite{kautz1964nonrandom} to construct a $d$-disjunct matrix. Let $\eta$ be an integer satisfying $2^\eta < 2 \mathrm{e}^{W(\frac{1}{2} d\ln{N})} < 2^{\eta + 1}$. Choose $C_{\mathrm{out}}$ as an $[n = q-1, r]_q$-RS code, where\begin{equation}
q = \begin{cases}
2 \mathrm{e}^{W(\frac{1}{2} d\ln{N})} = \frac{d \ln{N}}{W \left( \frac{1}{2} d \ln{N} \right)} &\parbox[t]{0.2\textwidth}{if $2 \mathrm{e}^{W(\frac{1}{2} d\ln{N})}$ is the power of 2.} \\
2^{\eta + 1}, &\parbox[t]{0.2\textwidth}{otherwise.}
\end{cases} \label{q}
\end{equation}
\noindent
Set $r = \left\lceil \frac{q - 2}{d} \right\rceil$, and let $C_{\mathrm{in}}$ be a $q \times q$ identity matrix. The complexity of $q$ is $\Theta \left( \mathrm{e}^{W(d\ln{N})} \right) = \Theta \left( \frac{d \ln{N}}{W \left( d \ln{N} \right)} \right)$ in both cases because
\begin{align}
& 2 \mathrm{e}^{W(\frac{1}{2} d\ln{N})} =  \frac{d \ln{N}}{W \left( \frac{1}{2} d \ln{N} \right)} \leq q < 2 \cdot 2 \mathrm{e}^{W(\frac{1}{2} d\ln{N})} =  \frac{2d \ln{N}}{W \left( \frac{1}{2} d \ln{N} \right)}. \nonumber
\end{align}

Let $C = C_{\mathrm{out}} \circ C_{\mathrm{in}}$. We are going to prove that $\mathcal{M} = \mathcal{M}_C$ is $d$-disjunct for such $q$ and $r$. It is well known~\cite{kautz1964nonrandom} that if $d \leq \frac{q - 1 - 1}{r - 1}$, $\mathcal{M}$ is $d$-disjunct with $t = q(q-1)$ tests. Indeed, we have
\begin{eqnarray}
\frac{q - 1 - 1}{r - 1} = \frac{q - 2}{\lceil \frac{q-2}{d} \rceil - 1} \geq  \frac{q-2}{\frac{q-2}{d} + 1 - 1} = d.
\end{eqnarray}

Since $q = O \left( \frac{d \ln{N}}{W \left( d \ln{N} \right)} \right)$, the number of tests in $\mathcal{M}$ is
\begin{align}
t = q(q-1) &= O \left(\frac{d^2 \ln^2{N}}{(W(d\ln{N}))^2} \right) \nonumber \\
&= O \left( \frac{d^2 \ln^2{N}}{ \left(\ln(d\ln{N}) - \ln{\ln(d\ln{N})} \right)^2} \right) \nonumber \\
&= O \left(\frac{d^2 \log^2{N}}{(\log(d\log{N}) - \log{\log(d\log{N})})^2} \right), \nonumber
\end{align}
because $\ln{x} - \ln{\ln{x}} \leq W(x) \leq \ln{x} - \frac{1}{2} \ln{\ln{x}}$ for any $x \geq \mathrm{e}$. Since $C_{\mathrm{out}}$ is an $[n, r]_q$-RS code, each of its codewords can be computed~\cite{von1997exponentiation} in time
\begin{align}
O(r^2) &= O \left( \left( \frac{\ln{N}}{\ln \left( d \ln{N} \right) - \ln \ln \left( d \ln{N} \right)} \right)^2 \right) = O \left( \frac{t}{d^2} \right) = O(t), \notag
\end{align}
and space
\begin{align}
S_1 &= O(r \log{q}/\log^2{r}) = O(q \log{q}) = O \left( d \ln{N} \right) = O(t). \label{S1}
\end{align}

Our task is now to prove that the number of columns in $\mathcal{M}_C$, i.e., $q^r$, is at least $N$. The range of $\frac{d\ln{N}}{W \left( \frac{1}{2}d \ln{N} \right)} \leq q < \frac{2d\ln{N}}{W \left( \frac{1}{2}d \ln{N} \right)}$ is:
\begin{eqnarray}
d + 2 &<& \frac{d\ln{N}}{\ln \left( \frac{1}{2}d \ln{N} \right) - \frac{1}{2} \ln \ln \left( \frac{1}{2}d \ln{N} \right)} \leq q \label{q1} \\
q &\leq& \frac{2d\ln{N}}{\ln \left( \frac{1}{2}d \ln{N} \right) - \ln \ln \left( \frac{1}{2}d \ln{N} \right)} < 2d\ln{N}. \label{q2}
\end{eqnarray}

These inequalities were obtained because $\ln{x} - \ln{\ln{x}} \leq W(x) \leq \ln{x} - \frac{1}{2} \ln{\ln{x}}$ for any $x \geq \mathrm{e}$. Then we have:
\begin{align}
q^{(q-2)/d} &= \left( \frac{q^q}{q^2} \right)^{1/d} \nonumber \\
&\geq \left( \frac{1}{q^2} \times \left( 2 \mathrm{e}^{W(\frac{1}{2} d\ln{N})} \right)^{q} \right)^{1/d} = \left( \frac{2^q}{q^2} \times \left( \mathrm{e}^{W(\frac{1}{2} d\ln{N})} \right)^{q} \right)^{1/d} \nonumber \\
&\geq \left( \frac{2^q}{q^2} \times \left( \mathrm{e}^{W(\frac{1}{2} d\ln{N}) \times 2 \mathrm{e}^{W(\frac{1}{2} d\ln{N})}}  \right) \right)^{1/d} = \left( \frac{2^q}{q^2} \times \mathrm{e}^{2 \times \frac{1}{2} d\ln{N}} \right)^{1/d} \label{LambertW} \\
&\geq N \times \left( \frac{2^q}{q^2} \right)^{1/d} > N. \label{lastN}
\end{align}
Equation (\ref{LambertW}) is achieved because $W(x)\mathrm{e}^{W(x)} = x$. Equation (\ref{lastN}) is obtained because $\left( \frac{2^q}{q^2} \right)^{1/d} \geq 1$ for any $q \geq 5$. Since $\frac{q - 2}{d} \leq r = \lceil \frac{q - 2}{d} \rceil < \frac{q - 2}{d} + 1$, the number of codewords in $C_{\mathrm{out}}$ is:
\begin{eqnarray}
N < q^{(q-2)/d} \leq q^r &<& q^{(q-2)/d + 1} = q \times q^{(q-2)/d} \label{N1} \\
&<& \frac{d\ln{N}}{W \left( \frac{1}{2}d \ln{N} \right)} \left( \frac{2^q}{q^2} \right)^{1/d} \times N. \label{N2}
\end{eqnarray}

Equation (\ref{N1}) indicates that the number of columns in $\mathcal{M}_{C}$ is more than $N$. To obtain a $t \times N$ matrix, one simply removes $q^r - N$ columns from $\mathcal{M}_{C}$. The maximum number of columns that can be removed is $O(d \ln{N} \times N^2)$ because of~(\ref{N2}).

\textbf{Decoding:} Consider the ratio $\frac{q-1}{r}$ implied by list size $d^\prime = \left\lceil \frac{q-1}{r} \right\rceil - 1 = \left\lceil \frac{q-1}{\lceil (q-2)/d \rceil} \right\rceil - 1$ of $[q-1, r]_q$-RS code. Parameter $d^\prime$ is also the maximum number of defective items that $\mathcal{M}$ can be used to identify because of Theorem~\ref{thr:conateCode}. We thus have
\begin{align}
d^\prime =  \left\lceil \frac{q-1}{\lceil (q-2)/d \rceil} \right\rceil - 1 \geq d \left(1 - \frac{d-1}{q + d - 2} \right) > \frac{d}{2}, \nonumber
\end{align}
because $q + d - 2 \geq 2d > 2(d-1)$. Since $d^\prime$ is an integer, $d^\prime \geq \left\lfloor \frac{d}{2} \right\rfloor + 1$.

Next we prove that $d^\prime = d - 1$ when $d$ is the power of 2, e.g., $d = 2^x$ for some positive integer $x$. Since $q$ is also the power of 2 as shown by (\ref{q}), suppose that $q = 2^{y}$ for some positive integer $y$. Because $q > d$ in (\ref{q1}), $2^y > 2^{x}$. Then $r = \lceil \frac{q-1}{d} \rceil = 2^{y - x}$. Therefore, $d^\prime = \left\lceil \frac{q-1}{r} \right\rceil - 1 = 2^x - 1 = d - 1$.

The decoding complexity of our proposed scheme is analyzed here. We have:
\begin{itemize}
\item Code $C_{\mathrm{out}}$ is an $(d^\prime = \left\lceil \frac{q-1}{\lceil (q-2)/d \rceil} \right\rceil - 1, L = O\left( \frac{n^4}{r^2} \right) = O(q^2 d^2))$-list recoverable code as in Theorem~\ref{thr:decodeListRS}. It can be decoded in time
\begin{align}
T_1 &= O(n^{3.57} r^{2.69}) = O \left( \frac{d^{3.57} \log^{6.26}{N}}{(\log(d\log{N}) - \log{\log(d\log{N})})^{6.26}} \right). \notag
\end{align}
Moreover, any codeword in $C_{\mathrm{out}}$ can be computed in time $O(r^2) = O \left( \frac{t}{d^2} \right)$ and space $S_1 = O(t)$ as in (\ref{S1}).
\item $C_{\mathrm{in}}$ is a $q \times q$ identity matrix. Then $\mathcal{M}_{C_{\mathrm{in}}}$ is a $q$-disjunct matrix. Since $d^\prime \leq d < q$, $\mathcal{M}_{C_{\mathrm{in}}}$ is also a $d^\prime$-disjunct matrix. It can be decoded in time $T_2 = d^\prime q$ and each codeword can be computed in space $S_2 = \log{q}$.
\end{itemize}

From Theorem~\ref{thr:conateCode}, given any outcome produced by at most $d^\prime$ defective items, those items can be identified in time 
\begin{align}
T_s &= nT_2 + T_1 + O(Lt) \notag \\
&= n d^\prime q + O \left( \frac{d^{3.57} \log^{6.26}{N}}{ \left( \log(d \log{N}) - \log{\log(d \log{N})} \right)^{6.26}} \right) + O \left( \frac{d^6 \log^4{N}}{\left( \log(d \log{N}) - \log{\log(d \log{N})} \right)^4} \right) \nonumber \\
&= O \left( \frac{d^{3.57} \log^{6.26}{N}}{\left( \log(d \log{N}) - \log{\log(d \log{N})} \right)^{6.26}} \right) + O \left( \frac{d^6 \log^4{N}}{\left( \log(d \log{N}) - \log{\log(d \log{N})} \right)^4} \right). \label{TotalTime}
\end{align}

Moreover, each entry (column) in $\mathcal{M}$ can be computed in time $O(t)$ ($O(tq) = O(t^{3/2})$) and space $O(\log{t} + \log{N}) + O(\max \{S_1, S_2 \}) = O \left( d \log{N} \right) = O(t)$ ($O(tq) = O(t^{3/2})$).
\end{proof}

If we substitute $d$ by $2^{\left\lfloor \log_2{d} \right\rfloor + 1}$ in the theorem above, the measurement matrix is $2^{\left\lfloor \log_2{d} \right\rfloor + 1}$-disjunct. Therefore, it can be used to identify at most $d^\prime = 2^{\left\lfloor \log_2{d} \right\rfloor + 1} - 1 \geq d$ defective items. The number of tests and the decoding complexity in the theorem remain unchanged because $d <  2^{\left\lfloor \log_2{d} \right\rfloor + 1} \leq 2d$.

\section{Evaluation}
\label{sec:eval}

We evaluated variations of our proposed scheme by simulation using $d = 2, 2^3, 2^7, 2^{10}, 2^{12}$ and $N = 2^{20}, 2^{40}, 2^{60}, 2^{80}, 2^{100}$ in Matlab R2015a on an HP Compaq Pro 8300SF desktop PC with a 3.4-GHz Intel Core i7-3770 processor and 16-GB memory.

\subsection{Numerical settings for $N, d$, and $q$}
\label{sub:dq}

We focused on nonrandom construction of a $t \times N$ $d$-disjunct matrix $\mathcal{M}$ for which the time to generate an entry is $\poly(t)$. Given integers $d$ and $N$, an $[n = q-1, r]_q$ code $C_{\mathrm{out}}$ and a $q \times q$ identity matrix $C_{\mathrm{in}}$ were set up to create $\mathcal{M} = \mathcal{M}_{C_{\mathrm{out}} \circ C_{\mathrm{in}}}$. The precise formulas for $q, r, t$ are $q = 2 \mathrm{e}^{W(\frac{1}{2} d\ln{N})}$ or $q = 2^{\eta + 1}$ as in (\ref{q}), $r = \lceil \frac{q - 2}{d} \rceil$, and $t = q(q-1)$. Note that the integer $q$ is the power of 2. Moreover, $N^\prime = q^r$ is the maximum number of items such that the resulting $t \times N^\prime$ matrix generated from this RS code was still $d$-disjunct. Parameter $d^\prime = \left\lceil \frac{q-1}{r} \right\rceil - 1 = \left\lceil \frac{q-1}{\lceil (q-2)/d \rceil} \right\rceil - 1$ is the maximum number of defective items that matrix $\mathcal{M}$ could be used to identify. The parameters $t_2 = 4800 d^2 \log{N}$ and $t_1 = d\log{N} (d\log{N} - 1)$ are the number of tests from Theorems~\ref{thr:StronglyExplicit} and~\ref{thr:nonrandom}. The numerical results are shown in Table~\ref{tbl:setting}.

\begin{center}
\begin{table*}[ht]
\centering
\caption{Parameter settings for $[q-1, r]_q$-RS code and resulting $q(q-1) \times N$ $d$-disjunct matrix: number of items $N$, maximum number of defective items $d$, alphabet size $q$ as in (\ref{q}), number of tests $t = q(q-1)$, dimension $r = \lceil \frac{q - 2}{d} \rceil$. Parameter $d^\prime = \left\lceil \frac{q-1}{\lceil (q-2)/d \rceil} \right\rceil - 1$ is the maximum number of defective items that the $t \times N$ resulting matrix can be used to identify. Parameter $N^\prime = q^r$ is maximum number of items such that resulting $q(q-1) \times N^\prime$ matrix generated from this RS code is still $d$-disjunct. Parameters $t_2 = 4800 d^2 \log{N}$ and $t_1 = d\log{N} (d\log{N} - 1)$ are number of tests from Theorems~\ref{thr:StronglyExplicit} and~\ref{thr:nonrandom}.}
\scalebox{1.1}{
\begin{tabular}{|l|c|c|l|c|c|c|r|r|}
\hline
$\qquad \ d$ & $N$ & $q$ & $\ t = q(q-1)$ & $r$ & $d^\prime$ & $N^\prime$ & \begin{tabular}{@{}c@{}} $t_1 =$ \\ $d\log{N} (d\log{N} - 1)$ \end{tabular} & \begin{tabular}{@{}c@{}} $t_2 = 4800 d^2 \log{N}$ \end{tabular} \\
\hline
\multirow{5}{*}{$2^3 = 8$} & $2^{20}$ & $2^6 = 64$ & $4,032$ & $8$ & $d-1$ & $2^{48}$ & $25,440$ & $6,144,000$ \\
						   & $2^{40}$ & $2^7 = 128$ & $16,256$ & $16$ & $d-1$ & $2^{102}$ & $102,080$ & $12,288,000$ \\
   						   & $2^{60}$ & $2^7 = 128$ & $16,256$ & $16$ & $d-1$ & $2^{102}$ & $229,920$ & $18,432,000$ \\
 						   & $2^{80}$ & $2^7 = 128$ & $16,256$ & $16$ & $d-1$ & $2^{102}$ & $408,960$ & $24,576,000$ \\
 						   & $2^{100}$ & $2^8 = 256$ & $65,280$ & $32$ & $d-1$ & $2^{256}$ & $639,200$ & $30,720,000$ \\
\hline
\multirow{5}{*}{$2^7 = 128$} & $2^{20}$  & $2^9 = 512$      & $261,632$   & $4$  & $d-1$ & $2^{36}$ & $6,551,040$ & $1,572,864,000$ \\
						     & $2^{40}$  & $2^{10} = 1,024$ & $1,047,552$ & $8$  & $d-1$ & $2^{80}$ & $26,209,280$ & $3,145,728,000$ \\
						     & $2^{60}$  & $2^{10} = 1,024$ & $1,047,552$ & $8$  & $d-1$ & $2^{80}$ & $58,974,720$ & $4,718,592,000$ \\
 						     & $2^{80}$  & $2^{11} = 2,048$ & $4,192,256$ & $16$ & $d-1$ & $2^{176}$ & $104,847,360$ & $6,291,456,000$ \\
 						     & $2^{100}$ & $2^{11} = 2,048$ & $4,192,256$ & $16$ & $d-1$ & $2^{176}$ & $163,827,200$ & $7,864,320,000$ \\
\hline
\multirow{5}{*}{$2^{10} = 1,024$} & $2^{20}$ & $2^{11} = 2,048$ & $4,192,256$ & $2$ & $d-1$ & $2^{22}$ & $419,409,920$ & $100,663,296,000$ \\
 								  & $2^{40}$ & $2^{12} = 4,096$ & $16,773,120$ & $4$ & $d-1$ & $2^{48}$ & $1,677,680,640$ & $201,326,592,000$ \\
	   							  & $2^{60}$ & $2^{13} = 8,192$ & $67,100,672$ & $8$ & $d-1$ & $2^{104}$ & $3,774,812,160$ & $301,989,888,000$ \\
						    	  & $2^{80}$ & $2^{13} = 8,192$ & $67,100,672$ & $8$ & $d-1$ & $2^{104}$ & $6,710,804,480$ & $402,653,184,000$ \\
						    	  & $2^{100}$ & $2^{14} = 16,384$ & $268,419,072$ & $16$ & $d-1$ & $2^{224}$ & $10,485,657,600$ & $503,316,480,000$\\
\hline
\multirow{5}{*}{$2^{12} = 4,096$} & $2^{20}$ & $2^{13} = 8,192$ & $67,100,672$ & $2$ & $d-1$ & $2^{26}$ & $6,710,804,480$ & $1,610,612,736,000$ \\
								  & $2^{40}$ & $2^{14} = 16,384$ & $268,419,072$ & $4$ & $d-1$ & $2^{56}$ & $26,843,381,760$ & $3,221,225,472,000$ \\
								  & $2^{60}$ & $2^{15} = 32,768$ & $1,072,398,336$ & $8$ & $d-1$ & $2^{120}$ & $60,397,731,840$ & $4,831,838,208,000$ \\
								  & $2^{80}$ & $2^{15} = 32,768$ & $1,072,398,336$ & $8$ & $d-1$ & $2^{120}$ & $107,373,854,720$ & $6,442,450,944,000$ \\
   							      & $2^{100}$ & $2^{15} = 32,768$ & $1,072,398,336$ & $8$ & $d-1$ & $2^{120}$ & $167,771,750,400$ & $8,053,063,680,000$ \\
\hline
\end{tabular}} 								

\label{tbl:setting}
\end{table*}
\end{center}

Since the number of tests from Theorem~\ref{thr:StronglyExplicit} is $O(d^2 \log{N})$, it \textit{should} be smaller than the number of tests in Theorem~\ref{thr:nonrandom}, which is $t = O(d^2 \log^2{N})$, and Theorem~\ref{thr:mainNonrandom}, which is $t = O \left(\frac{d^2 \log^2{N}}{(\log(d\log{N}) - \log{\log(d\log{N})})^2} \right)$. However, the numerical results in Table~\ref{tbl:setting} show the opposite. Even when $d = 2^{12} \approx 0.4\%$ of $N$, the number of tests from Theorem~\ref{thr:StronglyExplicit} was the largest. Moreover, there was no efficient construction scheme associated with it. The main reason is that the multiplicity of $O(d^2 \log{N})$ is $4,800$, which is quite large. Figure~\ref{fig:NOtests} shows the ratio between the number of tests from Theorem~\ref{thr:StronglyExplicit} and the number from Theorem~\ref{thr:mainNonrandom} (our proposed scheme) and between the number from Theorem~\ref{thr:nonrandom} and the number from Theorem~\ref{thr:mainNonrandom} (our proposed scheme). The number of tests with our proposed scheme was clearly smaller than with the existing schemes, even when $N = 2^{100}$. This indicates that the matrices generated from Theorem~\ref{thr:StronglyExplicit} and Theorem~\ref{thr:nonrandom} are \textit{good in theoretical analysis} but \textit{bad in practice.}

\begin{figure}[ht]
\centering
  \includegraphics[scale=0.25]{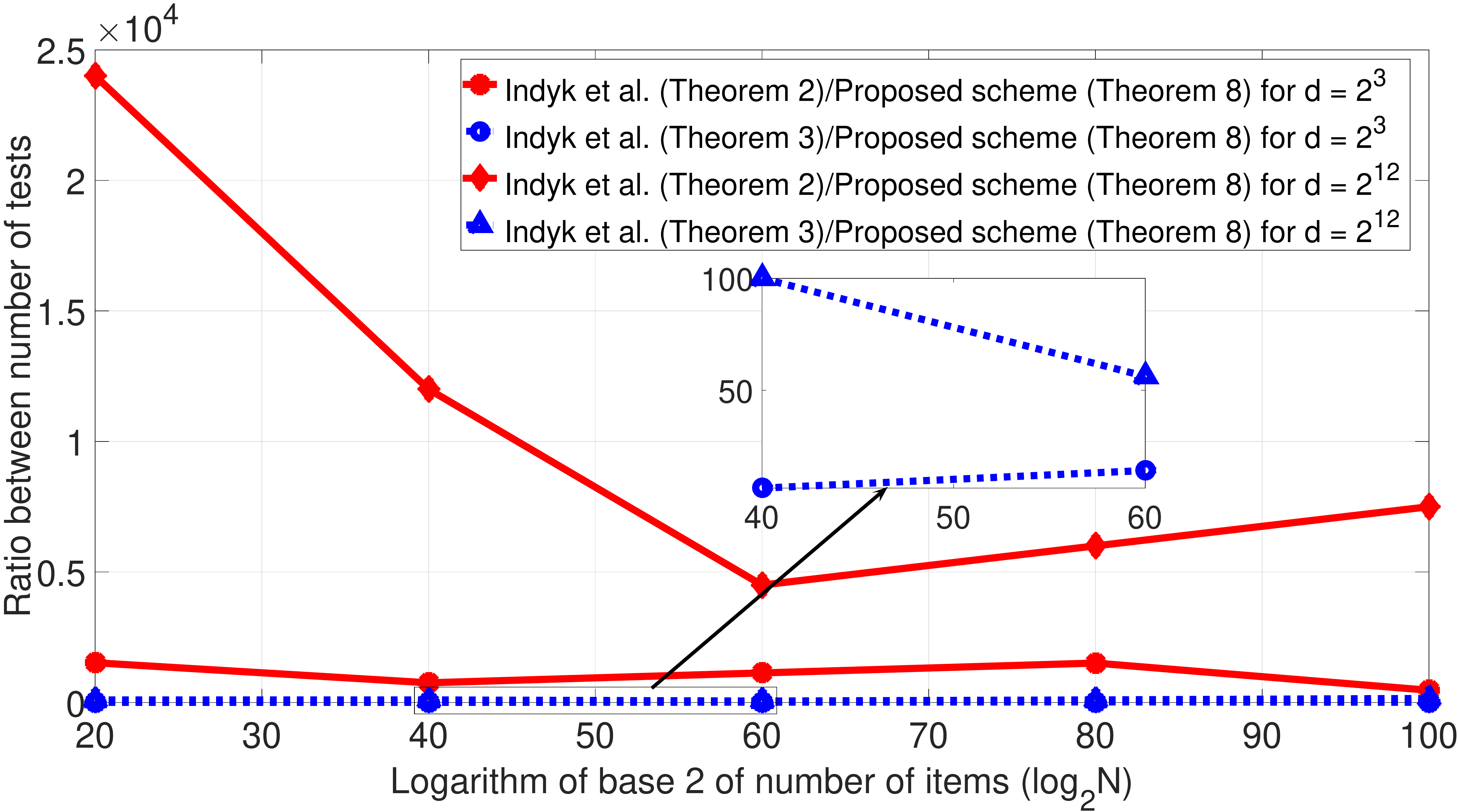}
  \caption{Ratio of number of tests from Theorem~\ref{thr:StronglyExplicit} and number from Theorem~\ref{thr:nonrandom} to number with proposed scheme (Theorem~\ref{thr:mainNonrandom}) for $d = 2^3, 2^{12}$ and $N = 2^{20}, 2^{40}, 2^{60}, 2^{80}, 2^{100}$. Ratio was always larger than 1; i.e., the number of tests in the proposed scheme is smaller than the compared one.}
  \label{fig:NOtests} 
\end{figure}

In contrast, a nonrandom $d$-disjunct matrix is easily generated from Theorem~\ref{thr:mainNonrandom}. It also can be used to identify at most $d-1$ defective items. If we want to identify up to $d$ defective items, we must generate a nonrandom $(d + 1)$-disjunct matrix in which the number of tests is still smaller than $t_1$ and $t_2$. Since the number of tests from Theorem~\ref{thr:mainNonrandom} is the lowest, its decoding time is the shortest. In short, for implementation, we recommend using the nonrandom construction in Theorem~\ref{thr:mainNonrandom}.

\subsection{Experimental results}
\label{sub:exprmt}

Since the time to generate a measurement matrix entry would be too long if it were $O(tN)$, we focus on implementing the methods for which the time to generate a measurement matrix entry is $\poly(t)$, i.e., $\langle \mathbf{3} \rangle, \langle \mathbf{4} \rangle, \langle \mathbf{8} \rangle, \langle 9 \rangle, \langle \mathbf{10} \rangle$ in Table~\ref{tbl:cmp}. However, to incorporate a measurement matrix into applications, random constructions are not preferable. Therefore, we focus on nonrandom constructions. Since we are unable to program decoding of list-recoverable codes because it requires knowledge of algebra, finite field, linear algebra, and probability. We therefore tested our proposed scheme by implementing $\langle \mathbf{4} \rangle$ (Theorem~\ref{thr:nonrandom2}) and $\langle \mathbf{8} \rangle$ (Corollary~\ref{cor:nonrandom}). This is reasonable because, as analyzed in section~\ref{sub:dq}, the number of tests in Theorem~\ref{thr:mainNonrandom} is the best for implementing nonrandom constructions. Since Corollary~\ref{cor:nonrandom} is derived from Theorem~\ref{thr:mainNonrandom}, its decoding time should be the best for implementation.

We ran experiments for $d = 2$ from Theorem~\ref{thr:nonrandom2} and $d = 2^3, 2^7$ from Corollary~\ref{cor:nonrandom}. We did not run any for $d = 2^{10}, 2^{12}$ because there was not enough memory in our set up (more than 100 GB of RAM is needed). The decoding time was calculated in seconds and averaged over 100 runs. When $d = 2$, the decoding time was less than 1ms. As shown in Figure~\ref{fig:DecodingTime}, the decoding time was linearly related to the number of tests, which confirms our theoretical analysis. Moreover, defective items were identified extremely quickly (less than $16$s) even when $N = 2^{100}$. The accuracy was always 1; i.e., all defective items were identified.

\begin{figure}[ht]
\centering
  \includegraphics[scale=0.25]{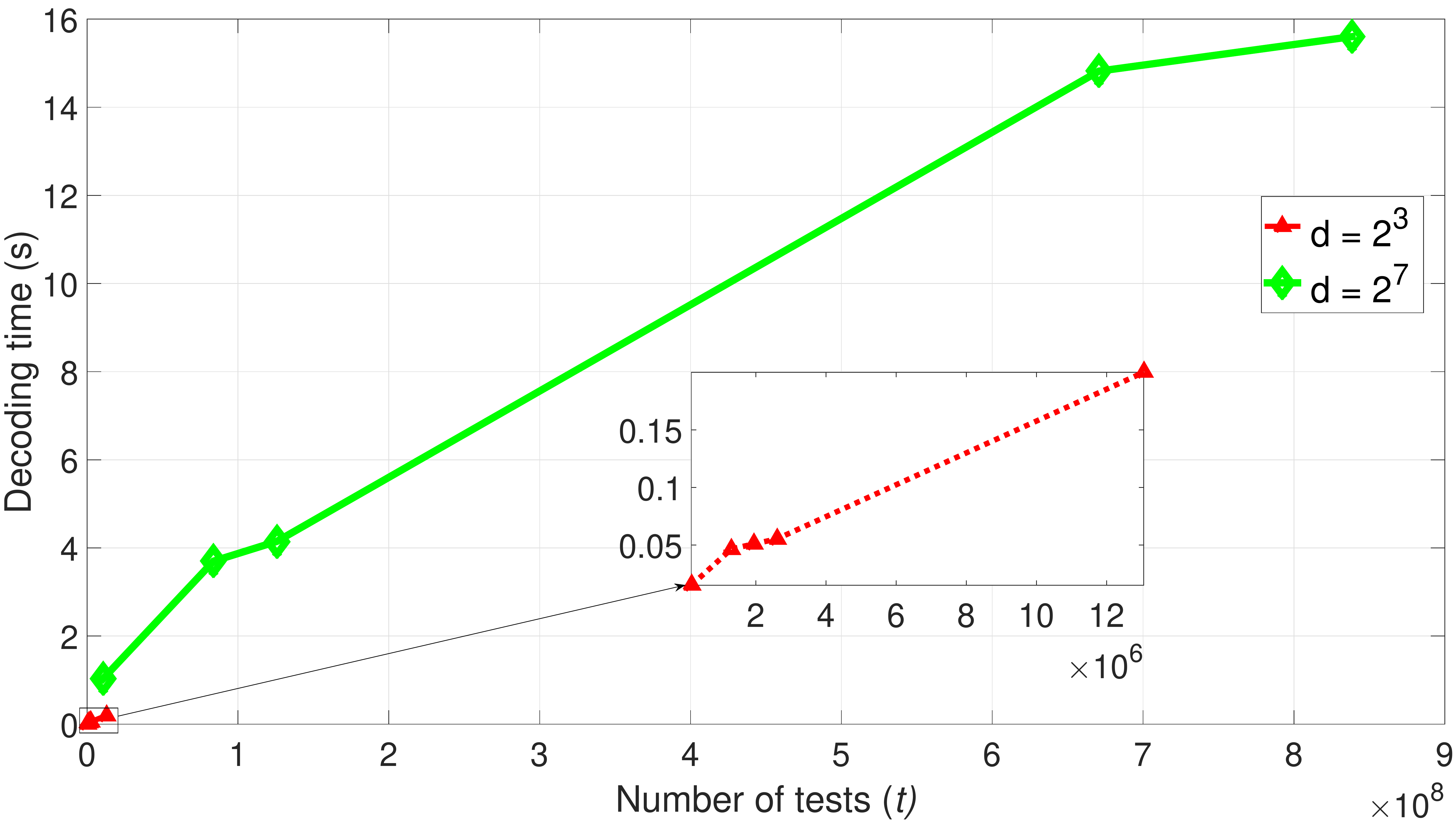}
  \caption{Decoding time for $d = 2^3$ and $d = 2^7$ from Theorem~\ref{thr:nonrandom2}. Number of items $N$ was $\{ 2^{20}, 2^{40}, 2^{60}, 2^{80}$, or $2^{100} \}$.}
  \label{fig:DecodingTime}
\end{figure}

\section{Conclusion}
\label{sec:cls}
We have presented a scheme that enables a larger measurement matrix built from a given $t \times N$ measurement matrix to be decoded in time $O(t \log{N})$ and a construction of a nonrandom $d$-disjunct matrix with $t = O \left(\frac{d^2 \log^2{N}}{(\log(d\log{N}) - \log{\log(d\log{N})})^2} \right)$ tests. This number of tests indicates that the upper bound for nonrandom construction is no longer $O(d^2 \log^2{N})$. Although the number of tests with our proposed schemes is not optimal in term of theoretical analysis, it is good enough for implementation. In particular, the decoding time is less than $16$ seconds even when $d = 2^{7} = 128$ and $N = 2^{100}$. Moreover, in nonrandom constructions, there is no need to store a measurement matrix because each column in the matrix can be generated efficiently.

\textit{Open problem:} Our finding that $N$ is become much smaller than $N^\prime$ as $q$ increases (Table~\ref{tbl:setting}) is quite interesting. Our hypothesis is that the number of tests needed may be smaller than $2 \mathrm{e}^{W(\frac{1}{2} d\ln{N})} \left( 2 \mathrm{e}^{W(\frac{1}{2} d\ln{N})} -1 \right)$. If this is indeed true, it paves the way toward achieving a very efficient construction and a shorter decoding time without using randomness. An interesting question to answer the question is whether there exists a $t \times N$ $d$-disjunct matrix with $t \leq 2 \mathrm{e}^{W(\frac{1}{2} d\ln{N})} \left( 2 \mathrm{e}^{W(\frac{1}{2} d\ln{N})} -1 \right)$ that can be constructed in time $O(tN)$ with each entry generated in time (and space) $\poly(t)$ and with a decoding time of $O(t^2)$.

\section{Acknowledgments}
The first author would like to thank Dr. Mahdi Cheraghchi, Imperial College London, UK for his fruitful discussions.

\bibliographystyle{ieeetr}
\bibliography{bibli}

\begin{thebibliography}{10}

\bibitem{dorfman1943detection}
R.~Dorfman, ``The detection of defective members of large populations,'' {\em
  The Annals of Mathematical Statistics}, vol.~14, no.~4, pp.~436--440, 1943.

\bibitem{damaschke2006threshold}
P.~Damaschke, ``Threshold group testing,'' in {\em General theory of
  information transfer and combinatorics}, pp.~707--718, Springer, 2006.

\bibitem{du2000combinatorial}
D.~Du, F.~K. Hwang, and F.~Hwang, {\em Combinatorial group testing and its
  applications}, vol.~12.
\newblock World Scientific, 2000.

\bibitem{d1982bounds}
A.~G. D'yachkov and V.~V. Rykov, ``Bounds on the length of disjunctive codes,''
  {\em Problemy Peredachi Informatsii}, vol.~18, no.~3, pp.~7--13, 1982.

\bibitem{atia2012boolean}
G.~K. Atia and V.~Saligrama, ``Boolean compressed sensing and noisy group
  testing,'' {\em IEEE Transactions on Information Theory}, vol.~58, no.~3,
  pp.~1880--1901, 2012.

\bibitem{cormode2005s}
G.~Cormode and S.~Muthukrishnan, ``What's hot and what's not: tracking most
  frequent items dynamically,'' {\em ACM Transactions on Database Systems
  (TODS)}, vol.~30, no.~1, pp.~249--278, 2005.

\bibitem{ngo2000survey}
H.~Q. Ngo and D.-Z. Du, ``A survey on combinatorial group testing algorithms
  with applications to dna library screening,'' {\em Discrete mathematical
  problems with medical applications}, vol.~55, pp.~171--182, 2000.

\bibitem{bui2018a}
T.~V. Bui, M.~Kuribayashi, M.~Cheraghchi, and I.~Echizen, ``A framework for
  generalized group testing with inhibitors and its potential application in
  neuroscience,'' {\em arXiv preprint arXiv:1810.01086}, 2018.

\bibitem{cai2013grotesque}
S.~Cai, M.~Jahangoshahi, M.~Bakshi, and S.~Jaggi, ``Grotesque: noisy group
  testing (quick and efficient),'' in {\em Communication, Control, and
  Computing (Allerton), 2013 51st Annual Allerton Conference on},
  pp.~1234--1241, IEEE, 2013.

\bibitem{lee2016saffron}
K.~Lee, R.~Pedarsani, and K.~Ramchandran, ``Saffron: A fast, efficient, and
  robust framework for group testing based on sparse-graph codes,'' in {\em
  Information Theory (ISIT), 2016 IEEE International Symposium on},
  pp.~2873--2877, IEEE, 2016.

\bibitem{kautz1964nonrandom}
W.~Kautz and R.~Singleton, ``Nonrandom binary superimposed codes,'' {\em IEEE
  Transactions on Information Theory}, vol.~10, no.~4, pp.~363--377, 1964.

\bibitem{indyk2010efficiently}
P.~Indyk, H.~Q. Ngo, and A.~Rudra, ``Efficiently decodable non-adaptive group
  testing,'' in {\em Proceedings of the twenty-first annual ACM-SIAM symposium
  on Discrete Algorithms}, pp.~1126--1142, SIAM, 2010.

\bibitem{guruswami2004linear}
V.~Guruswami and P.~Indyk, ``Linear-time list decoding in error-free
  settings,'' in {\em International Colloquium on Automata, Languages, and
  Programming}, pp.~695--707, Springer, 2004.

\bibitem{cheraghchi2013noise}
M.~Cheraghchi, ``Noise-resilient group testing: Limitations and
  constructions,'' {\em Discrete Applied Mathematics}, vol.~161, no.~1-2,
  pp.~81--95, 2013.

\bibitem{porat2008explicit}
E.~Porat and A.~Rothschild, ``Explicit non-adaptive combinatorial group testing
  schemes,'' in {\em International Colloquium on Automata, Languages, and
  Programming}, pp.~748--759, Springer, 2008.

\bibitem{guruswami2007algorithmic}
V.~Guruswami {\em et~al.}, ``Algorithmic results in list decoding,'' {\em
  Foundations and Trends{\textregistered} in Theoretical Computer Science},
  vol.~2, no.~2, pp.~107--195, 2007.

\bibitem{chowdhury2015faster}
M.~F. Chowdhury, C.-P. Jeannerod, V.~Neiger, E.~Schost, and G.~Villard,
  ``Faster algorithms for multivariate interpolation with multiplicities and
  simultaneous polynomial approximations,'' {\em IEEE Transactions on
  Information Theory}, vol.~61, no.~5, pp.~2370--2387, 2015.

\bibitem{von1997exponentiation}
J.~Von Zur~Gathen and M.~N{\"o}cker, ``Exponentiation in finite fields: theory
  and practice,'' in {\em International Symposium on Applied Algebra, Algebraic
  Algorithms, and Error-Correcting Codes}, pp.~88--113, Springer, 1997.

\bibitem{bui2018efficient}
T.~V. Bui, T.~Kojima, M.~Kuribayashi, and I.~Echizen, ``Efficient decoding
  schemes for noisy non-adaptive group testing when noise depends on number of
  items in test,'' {\em arXiv preprint arXiv:1803.06105}, 2018.

\bibitem{bui2017efficiently}
T.~V. Bui, M.~Kuribayashil, M.~Cheraghchi, and I.~Echizen, ``Efficiently
  decodable non-adaptive threshold group testing,'' in {\em 2018 IEEE
  International Symposium on Information Theory (ISIT)}, pp.~2584--2588, IEEE,
  2018.

\bibitem{parvaresh2005correcting}
F.~Parvaresh and A.~Vardy, ``Correcting errors beyond the guruswami-sudan
  radius in polynomial time,'' in {\em Foundations of Computer Science, 2005.
  FOCS 2005. 46th Annual IEEE Symposium on}, pp.~285--294, IEEE, 2005.

\bibitem{hoorfar2008inequalities}
A.~Hoorfar and M.~Hassani, ``Inequalities on the lambert w function and
  hyperpower function,'' {\em J. Inequal. Pure and Appl. Math}, vol.~9, no.~2,
  pp.~5--9, 2008.

\end{thebibliography}

\end{document}